\documentclass{article}

\usepackage{listings}
\usepackage{booktabs}   
\usepackage{subcaption} 
\usepackage {alltt}
\usepackage{mathpartir}
\usepackage{setspace}
\usepackage{amsmath}
\usepackage{amsthm}
\usepackage{amssymb}

\usepackage{multirow}
\usepackage{multicol}
\usepackage{mathtools}
\usepackage{url}
\usepackage{natbib}
\usepackage{hyperref}

\usepackage{tikz}
\usetikzlibrary{automata,positioning,arrows,calc,shapes}

\definecolor{lightYellow}{rgb}{1,0.98,0.92}
\definecolor{codegray}{rgb}{0.5,0.5,0.5}
\definecolor{codegreen}{rgb}{0,0.6,0}
\definecolor{codepurple}{rgb}{0.58,0,0.82}
\definecolor[named]{ACMPurple}{cmyk}{0.55,1,0,0.15}
\definecolor[named]{ACMGreen}{cmyk}{0.20,0,1,0.19}
\lstdefinestyle{nwyvern}{
  xleftmargin=\parindent,
  showstringspaces=false,
  basicstyle={\footnotesize \ttfamily},
  morekeywords={name, subtype, type, val, fun, def, new, let, in, if, then, else, assert},
  otherkeywords={=>},
  numberstyle=\tiny\color{codegray},
  commentstyle=\color{ACMGreen},
  stringstyle=\color{ACMPurple},
  captionpos=b,
  numbers=left,
  morecomment=[l][\color{ACMGreen}]{//},
  moredelim=[is][\underbar]{__}{__}
}

\lstdefinestyle{simple}{
  xleftmargin=\parindent,
  showstringspaces=false,
  basicstyle={\footnotesize \ttfamily},
  commentstyle=\ttfamily,
  morekeywords={name, subtype, type, val, fun, def, new, let, if},
  otherkeywords={<:,<=},
  captionpos=b,
  numbers=none,
}

\newcommand{\code}[1]{\texttt{#1}}

\newcommand{\subtypes}{<:\,}

\newcommand{\dcolon}{{::}}  

\newcommand{\SynaxKeyword}[1]{{\texttt{\textbf{#1}}}}

\newcommand{\MaterialAnnot}{\SynaxKeyword{$\cdot$}}
\newcommand{\ShapeAnnot}{\SynaxKeyword{@shape}}
\newcommand{\Name}[3]{\SynaxKeyword{name} \ #1 \ \{ #2 \Rightarrow #3 \}}
\newcommand{\Subtype}[3]{\SynaxKeyword{subtype} \ #1 \, #2 <: #3}
\newcommand{\Type}[3]{\SynaxKeyword{type} \ #1 \ #2 \ #3 }
\newcommand{\Val}[2]{\SynaxKeyword{val} \ #1 \ : \ #2 }
\newcommand{\Fun}[4]{\SynaxKeyword{def} \ #1(#3: #2): #4 }
\newcommand{\New}[3]{\SynaxKeyword{new} \ #1 \{ #2 \Rightarrow #3 \}}
\newcommand{\Let}[3]{\SynaxKeyword{let} \ #1 = #2 \ \SynaxKeyword{in} \ #3}
\newcommand{\DType}[2]{\SynaxKeyword{type} \ #1 \ = \ #2 }
\newcommand{\DVal}[3]{\SynaxKeyword{val} \ #1 \ : \ #2 \ = \ #3 }
\newcommand{\DFun}[5]{\SynaxKeyword{def} \ #1(#3: #2): #4 = #5}

\newcommand{\JLabel}[1]{\ \textsc{#1}}

\newcommand{\JJudge}[3]{\ensuremath{\Delta | \Sigma | #1 | #2 \vdash #3}}

\newcommand{\JType}[4]{\ensuremath{\Delta | \Sigma | #1 | #2 \vdash #3 : #4}}
\newcommand{\JSType}[3]{\JType{#1}{S}{#2}{#3}}
\newcommand{\JGSType}[2]{\JType{\Gamma}{S}{#1}{#2}}
\newcommand{\JGSMember}[3]{\Delta | \Sigma | \Gamma | S \vdash #1 \ni_{#2} #3}

\newcommand{\JGSWf}[1]{\Delta | \Sigma | \Gamma | S \vdash #1 \ \text{valid}}

\newcommand{\JGSExpose}[2]{\Delta | \Sigma | \Gamma | S \vdash #1 \Uparrow #2}
\newcommand{\JGSDowncast}[2]{\Delta | \Sigma | \Gamma | S \vdash #1 \searrow #2}
\newcommand{\JGSUpcast}[2]{\Delta | \Sigma | \Gamma | S \vdash #1 \nearrow #2}

\newcommand{\JSubtype}[4]{\Delta | \Sigma | #1 | #2 \vdash #3 \subtypes #4}
\newcommand{\JSSubtype}[3]{\JSubtype{#1}{S}{#2}{#3}}
\newcommand{\JGSSubtype}[2]{\JSSubtype{\Gamma}{#1}{#2}}
\newcommand{\JGxSubtype}[2]{\JSubtype{\Gamma}{\cdot}{#1}{#2}}

\newcommand{\Avoid}[4]{\Delta | \Sigma | \Gamma | S \vdash #1 \Uparrow^{#2}_{/#3} #4}
\newcommand{\JAvoid}[5]{\Delta | \Sigma | #1 | S \vdash #2 \Uparrow^{#3}_{/#4} #5}

\newcommand{\BoundJoin}[2]{#1 \sqcup #2}
\newcommand{\BoundMul}[2]{#1 \cdot #2}

\newcommand{\tmultiplier}[1]{\mathcal{M}_{\Delta\Sigma}(#1)}

\newcommand{\tadder}[1]{\mathcal{A}_{\Delta\Sigma}(#1)}

\newcommand{\energy}[1]{\mathcal{E}_{\Delta \Sigma \Gamma S}(#1)}

\newcommand{\StepsTo}[4]{#1 \ | \ #2 \ \Downarrow \ #3 \ | \ #4}
\newcommand{\StepsToN}[5]{#1 \ | \ #2 \ \Downarrow_{#3} \ #4 \ | \ #5}

\newcommand{\Rank}[2]{\mbox{rank}(#1, #2)}
\newcommand{\RankHead}[2]{\mbox{rank-head}(#1, #2)}

\tikzstyle{arrow}=[thick, ->, >=stealth]
\tikzset{squarednode/.style={rectangle, draw=black!60, fill=white!5, very thick, minimum size=5mm},roundnode/.style={circle, draw=black!60, fill=white!5, very thick, minimum size=5mm}}

\newtheorem{theorem}{Theorem}[section]
\newtheorem{lemma}{Lemma}[section]
\newtheorem{definition}[theorem]{Definition}
\newtheorem{example}[theorem]{Example}

\title{Semantically Separating Nominal Wyvern for Usability and Decidability}

\author{Yu Xiang Zhu$\ast$, Amos Robinson$\dagger$, Sophia Roshal$\ast$,\\
Timothy Mou$\ast$, Julian Mackay$\ddagger$ , Jonathan Aldrich$\ast$, Alex Potanin$\dagger$
~\footnote{$\ast$ Carnegie Mellon University, PA, USA; $\dagger$ Australian National University, Australia; $\ddagger$ Kry10}
}

\date{}

\begin{document}

\maketitle

\begin{abstract}
The Dependent Object Types (DOT) calculus incorporates concepts from functional languages (e.g. modules) with traditional object-oriented features (e.g. objects, subtyping) to achieve greater expressivity (e.g. F-bounded polymorphism).
However, this merger of paradigms comes at the cost of subtype decidability.
Recent work on bringing decidability to DOT has either sacrificed expressiveness or ease of use.
The unrestricted construction of recursive types and type bounds has made subtype decidability \textit{a much harder problem} than in traditional object-oriented programming.

Recognizing this, our paper introduces Nominal Wyvern, a DOT-like dependent type system that takes an alternative approach: instead of having a uniform structural syntax like DOT, Nominal Wyvern is designed around a ``semantic separation'' between \textit{the nominal declaration of recursive types} on the one hand, and \textit{the structural refinement of those types when they are used} on the other.
This \textit{design naturally guides the user} to avoid writing undecidably recursive structural types.

From a technical standpoint, this separation also makes guaranteeing decidability possible by allowing for an intuitive adaptation of material/shape separation, a technique for achieving subtype decidability by separating types responsible for subtyping constraints from types that represent concrete data.
The result is a type system with syntax and structure familiar to OOP users that achieves decidability without compromising the expressiveness of F-bounded polymorphism and module systems as they are used in practice.
\end{abstract}


\urlstyle{rm}

\section{Introduction}
\label{ch:intro}

The abstraction features of programming languages provide flexibility through their support for code reuse and modularization. However, while functions and methods are the agreed-upon way to abstract shared computation, the approach for how best to achieve type abstraction and polymorphism is less settled. On the one hand are pure object-oriented (OO) languages, where polymorphism is supported in the form of subtyping. In such languages, types are generally monomorphic except for the ability of specific types to act like general ones. On the other hand, are functional languages, where parametric polymorphism allows type variables to stand for concrete types.

Programming language researchers have long wanted to get the best of both worlds. The Dependent Object Types (DOT) calculus \cite{odersky:dot}, on which Scala is based, is one of the prime examples of merging the two paradigms. By combining the modules of functional languages with the objects of the OO world, DOT supports more expressive types, with very general forms of features such as bounded quantification, which enable programmers to state more precise static guarantees about programs.

However, the merger of functional and OO features is far from simple, since their features are deeply rooted in the philosophical and technical differences between structural and nominal types.  Traditionally, OO languages are nominal, in that types with different names are semantically different, and may occupy different places in the subtype hierarchy, the backbone of OO abstraction.
In contrast, functional languages have traditionally been more structural, meaning that the structure of a type is what defines the type, and the name is a mere convenience in referring to it. This difference is closely related to how abstraction is achieved in each paradigm, so the degree of nominality vs structurality is a key factor in designing a merger.


While merging these paradigms provides additional expressiveness, in DOT it came at the cost of decidable type checking \cite{hu:decidable_dsub_fragments}. Traditional subtyping and parametric polymorphism are well-studied and easily decidable, yet it is easy to get an undecidable system after combining the two. It was proved early on that bounded quantification is an undecidable problem \cite{pierce:bounded_quant_undecidable}. The reasons for undecidability are technically involved, but intuitively, a major issue is the ability to define types that subtype a type parameterized by themselves. This problem also manifests itself in practical languages that provide similar forms of quantification, including Scala \cite{amin:foundations_pdt} and Java \cite{grigore:java_turing_complete}.
Having an undecidable component in a language can pose real problems to the user. Programmers unfamiliar with the intricacies of their language's type system may end up being unable to compile their ``well-written'' program. The compiler's failure---which essentially comes in the form of a timeout---typically does not provide meaningful guidance on what the programmer should do to fix the problem.
That such problems are rare in practice is of modest comfort when you consider how severely they would impact the programmer if and when they do occur.

This paper presents Nominal Wyvern, a new core type system for the Wyvern programming language \cite{nistor:wyvern,kurilova_type-specific_2014,mackay_decidable_2019,potanin_defaulting_2004,lee_theory_2015,kurilova_type-specific_2014,kurilova_wyvern_2014}, which is itself inspired by DOT.
Nominal Wyvern takes an alternative route to merge structural and nominal types. DOT is fundamentally a structural calculus.
To encode nominal restrictions in DOT, one uses abstract type members: path-dependent types are referred to by name since they, combined with the bounds on type members, make up the nominal subtype hierarchy. The rest of the DOT type system is structural. Instead of adding more restrictions on top of this system, Nominal Wyvern shifts the balance away from structural types and toward nominal types. By introducing explicit nominality, we are able to \textit{distinguish between semantically-distinct occurrences of nominal and structural types}, which would use the same syntactic form in DOT.

In particular, Nominal Wyvern's nominal typing system introduces top-level type definitions, which separates the definition of a recursive type (such as \texttt{List}) from refinements of that type (such as a list of integers, which is written as \texttt{List \{ type Element = Int \}}). It also separates the declaration of a type member's bound (such as a type member declaration \texttt{type Element <=\ Comparable} within an ordered list type) from subtyping relationships between nominal types, which we now require to be explicitly declared at the top-level (written as \texttt{subtype List <:\ Collection}).\footnote{In Section~\ref{sec:increasing-expressivity:extension} we delve deeper into the expressivity implications and trade-offs that occur when adding further nominal types such as \texttt{IntList}.}
Both type declarations and subtyping are thus made more explicit and meaningful. Not only is this more natural to users familiar with traditional OO languages, the added nominality aids in separating the types that are designed to constrain recursive type definitions from the types that represent concrete data, a key step in achieving subtype decidability via material/shape separation \cite{greenman:shapes}. This built-in separation overcomes the need for additional, artificial restrictions to the type hierarchy, such as were needed in prior work~\cite{mackay:decidable_wyvern,mackay:popl2020}, resulting in a simpler formulation of material/shape separation rules.

The main contributions of this paper are 1) the design of a type system that combines nominal, recursive types with type members and type refinement; 2) a natural adaptation of material/shape separation to this type system, resulting in 3) guarantees of subtype decidability, without sacrificing 4) type safety, simplicity, and expressiveness. Section \ref{ch:background} discusses in detail the earlier research in DOT, subtype decidability, and nominality that motivated this paper. Section \ref{ch:design} presents the grammar and typing rules of Nominal Wyvern and explains how this design facilitates usability and decidability. Section \ref{ch:decidability} delves into the subtyping rules and outlines a proof of subtype decidability (the full proofs are in the supplementary material). Section \ref{ch:safety} goes over the type safety guarantees of Nominal Wyvern. Section \ref{ch:expressiveness} demonstrates the expressiveness of Nominal Wyvern by presenting several examples of common programming patterns in Nominal Wyvern syntax. Section \ref{ch:related} discusses related work and how Nominal Wyvern fits in. Finally, Section \ref{ch:conclusion} concludes the paper.

\section{Nominal Wyvern}
\label{ch:background}
\subsection{DOT and Path-Dependent Types}
\label{sec:bg:dot}

The dependent object types (DOT) calculus \cite{amin:foundations_pdt} was developed as a type-theoretic foundation for Scala. The key distinguishing feature of the DOT calculus is objects with type members. Unifying concepts from objects and modules allows DOT to model types that are dependent on objects.

\begin{lstlisting}[mathescape, style=nwyvern, label={lst:pdt_primer1}, caption={Path-dependent type}, language=Scala, float]
class Bank { b =>
  type Card
  def applyForCard(name: String): b.Card = ...
  def payOff(c: b.Card): Unit            = ...
}

val chase : Bank       = ...
val pnc   : Bank       = ...
val myCard: chase.Card = chase.applyForCard("freedom")
chase.payOff(myCard)    // OK
pnc.payOff(myCard)      // type mismatch
                        // found:    chase.Card
                        // required: pnc.Card
\end{lstlisting}

In Listing \ref{lst:pdt_primer1}, the \code{Bank} class defines an abstract type member \code{Card}. This means if you have a value, \code{chase}, of type \code{Bank}, calling \code{applyForCard} on it would return a value of type \code{chase.Card}. This is a path-dependent type because the type is not self-contained. It depends on another variable in the environment. Since the exact type of \code{chase.Card} is unknown (abstracted away in \code{Bank}), paying it off at any other Bank would not type check, even if the underlying \code{Card} type for the two Banks are the same. The type system thus allows the code to model the real-world restrictions of cards (one cannot pay off a card from one bank at another bank) without restricting the number of possible banks that can be dynamically created.

Unlike modules, a type member in an object
can be specified with bounds on its subtyping relation.
This provides the language with not only the ability to represent and typecheck traditional object/record or module types but also more expressive types that are related to each other.

\subsection{Subtyping and Undecidability}
\label{sec:bg:undecidable_subtyping}

Subtyping is a form of declaration-site inclusion polymorphism \cite{cardelli:types_abstraction_polymorphism} that allows one type to masquerade as another. It is characterized by the substitution principle: if S is a subtype of T (written S \subtypes T), then values of type S can act like values of type T. Clearly, this provides additional expressive power to programmers by enabling a limited form of bounded parametric polymorphism for functions even without traditional parametric polymorphism support.

Subtype checking is the procedure for checking if one type subtypes another. In the bank scenario above, subtype checking is easily decidable. Since all the types are self-contained names, the predefined subtyping relations 
define a partial order on all the type names. Subtype checking is thus checking if the two types are correctly related with respect to the partial order.

However, many modern object-oriented languages also support parametric polymorphism, either in the form of type parameters (e.g.\ Java generics, C++ templates) or type members (e.g. Scala type members).
As a result, not all types are predefined monolithic names anymore. The possibility of constructing new types (by filling in type parameters or refining type members) means subtype checking in these systems must evolve to be structural and recursive. For example, A \subtypes B does not necessarily mean that each type parameter/member of A is a subtype of the corresponding parameter/member of B. Figure \ref{fig:type_param_variance} illustrates that when A <: B, a type parameter/member can be covariant (i.e. it preserves this subtyping relation) or contravariant (i.e. it reverses this subtyping relation).


\begin{figure}[t]
\captionsetup[subfigure]{aboveskip=\smallskipamount,font=small}
\begin{subfigure}{.49\textwidth}
\begin{lstlisting}[mathescape, style=simple, language=Scala]
class ReadStream[+T] {
  def read(): T
}
val rs1 : ReadStream[Int] = ...
val rs2 : ReadStream[Num] = rs1
rs2.read()    // return type $\le$ Num
\end{lstlisting}
\caption{\emph{Covariant}(\code{+}) type member: Int \subtypes Num\\ $\Rightarrow$ ReadStream[Int] \subtypes ReadStream[Num]}
\end{subfigure}
\begin{subfigure}{.49\textwidth}
\begin{lstlisting}[mathescape, style=simple, language=Scala]
class WriteStream[-T] {
  def write(x: T): Unit
}
val ws1 : WriteStream[Num] = ...
val ws2 : WriteStream[Int] = ws1
ws2.write(1)  // input type $\ge$ Int
\end{lstlisting}
\caption{\emph{Contravariant}(\code{-}) type member: Int \subtypes Num\\ $\Rightarrow$ WriteStream[Num] \subtypes WriteStream[Int]}
\end{subfigure}
\caption{Type parameter variance, using Scala's type parameter syntax}
\label{fig:type_param_variance}
\end{figure}

Subtyping difficulties arise when a type S is defined as a subtype of some type parameterized with S itself. Such recursive definitions are used heavily in F-bounded polymorphism \cite{canning:f_bounded_poly}, a generalization of bounded polymorphism where the bounded type can appear in its bound. One common usage, shown in Figure \ref{fig:f_bounded_poly_ex}, is using recursive bounds to specify features of the bounded type. In this case, \code{String} is defined recursively so that functions expecting Cloneable objects can make more specific inferences about the return type of their \code{clone()} method.

\begin{figure}[t]
\begin{lstlisting}[mathescape, style=simple, language=Scala]
trait Cloneable[T] {
  def clone(): T
}
class String extends Cloneable[String] {
  def clone(): String = ...
}
def makeClone[T <: Cloneable[T]](x: T) = x.clone()
\end{lstlisting}
\caption{F-bounded polymorphism (in Scala)}
\label{fig:f_bounded_poly_ex}
\end{figure}

Subtype checking on these constructed types already involves recursively looking into the structure of both types to make sure all members/parameters satisfy the subtyping relation. Thus, recursive bounds are a potential cause for concern since subtype checking can now possibly loop back to a type that has been visited before. Prior research shows this is indeed a challenge \cite{grigore:java_turing_complete, amin:foundations_pdt}, and as mentioned before, this can cause major usability problems.  Modern compilers are expected to provide error messages that provide reasonable guidance about what can be done to the source code to fix the problem, but this is effectively impossible when the type system is undecidable.
By having a decidable system with clearly defined constraints, compilers and other programming tools will be able to much better assist programmers in expressing what they want.

\subsubsection{Getting Back Decidability}

Unfortunately, subtype checking in systems like F\textsubscript{<:} is not a simple case of cycle detection. \citet{ghelli:divergence_fsub} demonstrated that it is possible for the context to grow indefinitely as derivation progresses.
Implementing a simple looping detector is not necessarily useful since, similar to just adding a time-out in the compiler, it does not help the programmer fix the problem.
Many have since proposed enforcing some sort of subtype dependency restrictions so that infinitely looping derivations never occur. The most notable is the ban on ``expansive inheritance'' by \citet{kennedy:expansive_inheritance} since it is used by the Scala compiler.
An ``expansive edge'' exists when one type parameter of type S appears at a deeper nested level in the supertype of S. However, as the authors themselves acknowledged, this solution is not immediately applicable to Java wildcards. In addition, as \citet{greenman:shapes} pointed out, this restriction prevents a common pattern for expressing certain ``features'' of types: recall in Figure \ref{fig:f_bounded_poly_ex} the Cloneable type is used by its subtypes to signal they have a \code{clone()} method.
But if we want a generic list to be cloneable we would get the definition in Figure \ref{fig:cloneable_expansive_inheritance}, which now includes an expansive edge from E to E.


\begin{figure}[t]
\begin{lstlisting}[mathescape, style=simple, language=Scala]
class List[E] extends Cloneable[List[E]] {z =>
  def clone(): List[E] = ...
}
\end{lstlisting}
\caption{Cloneable list causing expansive inheritance}
\label{fig:cloneable_expansive_inheritance}
\end{figure}

\subsubsection{Material/Shape Separation}
\label{sec:bg:ms_separation}

The solution of Nominal Wyvern is adapted from the ``material/shape separation'' idea proposed by \citet{greenman:shapes} for Java-like languages (instead of DOT, which has type members). Material/shape separation is a conservative way of separating all types in a program into two camps: materials and shapes: Material types are concrete types that actually represent data, and are passed around in a program. Shape types, on the other hand, are only used to constrain other types via type bounds.  Shapes are typically parameterized by the types they bound, thus creating loops during subtyping.  The key insight in shape-material separation is that as long as only shapes create subtyping loops, and only materials are instantiated, these subtyping loops are harmless---they cannot cause the typechecker to loop. Revisiting the earlier cloneable example, the \code{Cloneable} type would be considered a shape since its main role is to bound other types, such as \code{String}. Conversely, \code{String} is an example of a material type, representing concrete objects with data (characters, in this case). As one can imagine, this restriction does limit the expressivity of the language. However, \citet{greenman:shapes} discovered that this restriction is completely viable in practice. After studying a large corpus of existing code (13.5 million lines of Java), the authors found that current coding practices already follow this separation almost universally and that the rare counterexamples can easily be made to conform to it.


In summary, the benefit of this solution is twofold: 1) the restriction is already compatible with industry programming standards, meaning that it does not require any major shift in programming practices for its adoption; and 2) the restriction is easy to understand and identify due to a limited number of intuitive uses of shapes. The restriction will be explained in detail in Section \ref{sec:design:ms_separation}.

\subsubsection{Material/Shape Separation for DOT}
\label{sec:bg:ms_separation_dot}



The original material/shape separation was formulated for Java generics, which use type parameters.
Achieving a similar separation for DOT-based languages requires adapting the material/shape separation to instead work on type members. The most straightforward translation from a type parameter-based system like Java to a type member-based system like DOT is using type members to represent type parameters.
However, \citet{greenman:shapes} restricts shapes so that they may not be used to instantiate type parameters.
If we were to directly translate this restriction to type members, we would end up disallowing the use of shapes when defining type member bounds.
Such a restriction on using shapes as type members is much stricter than it was originally intended for in Java, as type members are used for more than just specifying type parameters.
For example, recall our \code{Bank} example from Listing \ref{lst:pdt_primer1}, where we defined a bank type with an associated type member denoting the type of cards.
Suppose that we wanted to define a bank type with cards that we the user could clone, as in Listing \ref{lst:pdt_primer1_cloneable}.
Here, the \code{Card} member is required to be some subtype of the \code{Cloneable} shape that we saw earlier.
If we tried to apply the restriction from \citet{greenman:shapes} to type members, this program would be outlawed, and the utility of type members would be severely limited.


\begin{lstlisting}[mathescape, style=nwyvern, label={lst:pdt_primer1_cloneable}, caption={Cloneable bank (in Scala)}, language=Scala]
  class CloneableBank { b =>
    type Card <= Cloneable
    def applyForCard(name: String): b.Card = ...
    def payOff(c: b.Card): Unit            = ...
  }
\end{lstlisting}

Nominal Wyvern's adaptation of material/shape separation is inspired by work on Decidable Wyvern \cite{mackay:decidable_wyvern}, in which the author proposed an adaptation of material/shape separation to a DOT-based system. Their solution is a combination of semantic and syntactic restrictions that mimic the intent of the original separation rules. 
Nominal Wyvern differs from Decidable Wyvern in having a nominal typing and subtyping system. This allows for a simpler and more intuitive set of material/shape separation restrictions, as well as a simpler proof of subtype decidability. The difference and the benefits of Nominal Wyvern will be explained in detail in Section \ref{sec:design:ms_separation}.


Our work departs slightly from previous work on material/shape separation by requiring explicit shape annotations.
In \citet{mackay:decidable_wyvern}, the compiler  infers which types are shapes and which are materials.
Instead, we see material/shape annotations as aids in helping the user reason about their program: not only does it make the intention of a type clearer to any human readers and the compiler, but the act of writing annotations as part of the code also aids the programmer in being mindful of what role they want each type to play. As \citet{greenman:shapes} pointed out, this separation of roles is already prevalent in the industry, and a deviation is usually indicative of ``poor utilization of generics'' or other hacks due to poor language support. With this in place, the compiler can be much more useful in helping users identify illegal usages of shape types (which may be indicative of larger issues) by pinpointing the exact violating use.
This is in contrast to the existing undesirable behaviours of compilers for undecidable type systems, such as timing out due to an illegal recursive type.
On the other hand, even if a material type is erroneously used as a shape, the compiler is able to identify the circular dependency and all the types involved so that the user can locate from this subset of types which one(s) should be a shape instead. We describe how to compute such a cycle in Section \ref{sec:design:ms_separation}.

\subsection{Nominality}
\label{sec:bg:nominality}

The typing of the DOT calculus is already considered partly nominal, due to type ascriptions which can make a type member abstract from the perspective of clients \cite{rompf:dot_soundness}. Abstract type members created with upper bounds are uniquely identified by their name (including the paths leading up to their containing objects). However, structurality plays a significant role in DOT due to the flexibility to construct structural types freely.
This structurality has two downsides: first, the difficulty in applying material/shape separation, and second, a less rigid type hierarchy from a programmer's perspective. Nominal Wyvern is closer to the nominality of traditional OO languages such as Java than it is to DOT: Nominal Wyvern's syntax mandates all structures be named, and that the subtyping relations between these named structures be entirely nominal as well.
We believe that applying nominality to type members has resulted in a simpler and more usable system.
Compared to Decidable Wyvern \citep{mackay:decidable_wyvern}, the formulation of nominality in Nominal Wyvern also simplifies the material/shape separation rules; where Decidable Wyvern used a mutually recursive grammar to express the restrictions on how shapes can be refined and used, our restrictions are conceptually closer to those of a nominal system such as Java.

\subsubsection{Typing Nominality}

In Nominal Wyvern, we require all object types to be pre-declared and named; this restriction is in contrast to DOT, where a new object type can be defined anywhere anonymously by simply writing out its members. We impose this restriction for two reasons:
\begin{enumerate}
  \item Usability: Given the widespread success of nominal languages, such as Java and C++, we believe that giving names to structures makes the code easier to understand since the names would be representative of what each structure is intended for. This is especially important as the object gets larger and contains more kinds of members (i.e.\ type members, field members, function members). Named definitions can also improve error messages, as the offending types have a unique name which can be referred to.
  \item Runtime performance: Having named structures with fixed width allows the compiler and language runtime to more easily represent objects with a flat memory layout. This memory layout makes runtime object management simpler and faster.
\end{enumerate}

\subsubsection{Subtyping Nominality}
A nominal subtyping system requires subtyping relationships between types to be explicitly declared; such subtyping relationships are then checked for structural compatibility.
This explicit declaration of subtyping provides two similar benefits:

\begin{enumerate}
  \item Usability: Explicitly naming subtyping relations avoids accidental subtyping between types whose signatures happen to match.
  Explicitly declared subtyping relations can also provide better and localised error messages: in a system with implicit subtyping, type mismatches may be misreported as subtyping failures.
  \item Type-checking performance: Having all subtyping relations defined explicitly means most type checking can be done on the type declarations, rather than when types are used. The result of these checks can then be saved and reused throughout the typechecking of program expressions, saving repeated checks that may be long and recursive.
\end{enumerate}


\section{Grammar and type system}
\label{ch:design}

\subsection{Basics}
\label{sec:design:grammar}

Figure~\ref{fig:nominal_wyvern_grammar} 
introduces the formal grammar of Nominal Wyvern.
Programs in Nominal Wyvern consist of two parts: a set of top-level declarations ($\overline{D}$), and a main expression ($e$) to be evaluated during execution.

\begin{figure*}
\input{grammar}
\caption{Nominal Wyvern Grammar}
\label{fig:nominal_wyvern_grammar}
\end{figure*}

The set of top-level declarations defines all object types and their subtype relations. Each named type declaration binds a structure to a name. Similar to DOT, an object type can contain type members, fields, and method types. Type members in Nominal Wyvern are defined as either an upper bound on a type, a lower bound on a type, or an exact bound.
This definition of type members differs from DOT, which allows defining type members that range between a given upper and lower bound; this design decision was made to simplify the syntax since in most practical cases a bound on one side is sufficient. In Section \ref{ch:expressiveness} we will show that this restriction still allows enough expressive power to encode common patterns from both object-oriented languages and functional languages. Each explicit subtype declaration sets up a relation between two named types with an optional refinement on the LHS. Section \ref{sec:design:binary_typing} details how the top-level declarations differ from their DOT counterparts.

A type in Nominal Wyvern is made of a base type and a refinement. Base types define a type, while the refinement allows ad-hoc modifications to the base type that override the definitions of some of its members. A base type is either top ($\top$), bottom ($\bot$), a named type ($n$), or a path-dependent type ($p.t$).
In general path-dependent types such as pDOT \cite{rapoport:a_path_to_dot}, a path $p$ is a variable appended with zero or more successive field accesses; however, to simplify the metatheory we only consider paths with exactly one element, such as variables ($x$), as in Decidable Wyvern \cite{mackay:decidable_wyvern}.
For technical reasons explained in the type safety proof (Section \ref{ch:safety}), paths can also contain heap locations ($l$), but these do not occur in user-written programs.
Types and decidability are the main focus of this paper, so Nominal Wyvern only allows type member refinements; we do not consider method or field refinements in this presentation.

All expressions in Nominal Wyvern produce objects. To simplify our presentation and proofs, the grammar follows A-Normal Form, where method targets and arguments must both be single-length paths (variables $x$ or locations $l$). Any non-path object must be let-bound to a variable (or to a field of a variable) before its members can be accessed or be used as an argument to methods or field constructors. This restriction does not hinder expressivity; translating arbitrary expressions to this restricted form is standard. The \code{new} expression form creates objects. Since different names can correspond to the same structure, the exposed type of an object is specified during creation. Other expression forms such as \code{let} and method calls are standard.

Each named type declaration and type member are annotated with the material/shape type.
These annotations are designed to help programmers and compilers better understand the intention of types. These annotations are considered metadata external to the core grammar; they are only used by a separate material/shape separation check (Section \ref{sec:design:ms_separation}) done after parsing, before being discarded.
The type-checking rules do not use the annotations, so for clarity we omit them from the grammar wherever they are not relevant.

Figure \ref{fig:nominal_wyvern_syntax_eg} shows part of a Nominal Wyvern program that defines an immutable set container type.
The \code{Equatable} shape type defines an interface that can be used to check if two values are equal; its type member \code{EqT} defines the type of values that can be checked.
This type is an example of F-bounded polymorphism, similar to the earlier \code{Cloneable} example from Figure \ref{fig:f_bounded_poly_ex}.
The equality check is contravariant in the type member \code{EqT} as it accepts a \code{EqT} as an argument, so the type member is declared contravariantly with a lower bound of bottom ($\bot$).

Next, the \code{Fruit} type is declared as a material type.
Each \code{Fruit} has fields for its identifier and weight.
It also defines the associated type member \code{EqT} to equal \code{Fruit}, and defines an equality check method that takes another fruit and returns a boolean.
This type member and method matches the structure of the \code{Equatable} shape, so the subtyping declaration that states that \code{Fruit} is a subtype of the \code{Equatable} shape succeeds.

The \code{Set} type is declared as a material and declares a type member \code{ElemT}, which denotes the type of elements that can be stored in the set.
The set needs to be able to check equality of its elements, so the type of elements is constrained to be some subtype of the \code{Equatable} shape.
The definition of \code{Equatable} shape on its own means that an element can be compared to \emph{something}, but it does not necessarily mean that it can be compared to \emph{itself}.
To ensure that the element can be compared to itself, we need to constrain \code{Equatable}'s member type to be the actual element type.
We encode this constraint using the refinement \code{Equatable \{ type EqT = self.ElemT \}}.
In Java, we would write the same constraint using the generic type definition \code{class Set< ElemT extends Equatable<ElemT> >}.

The \code{Set} type also defines an \code{insert} method which takes an element (\code{self.ElemT}) and returns a new set with the same type of elements.

Finally, the main program constructs an \code{apple} of type \code{Fruit} and a set that can store fruit, before putting the apple in the set.
We omit the details of the expressions here, as our focus is on the type declarations and the decidability of subtyping.

\begin{figure}[t]
\begin{lstlisting}[mathescape, language=Scala, style=nwyvern, escapechar=|, basicstyle={\scriptsize \ttfamily}]
@shape name Equatable { self =>
  type EqT >=              $\bot$
  def equals(x: self.EqT): Bool
}

name Fruit { self =>
  val id:                  Int
  val weight:              Float
  type EqT =               Fruit
  def equals(x: Fruit):    Bool
}

subtype Fruit <: Equatable

name Set { self =>
  type ElemT <= Equatable { type EqT = self.ElemT }
  def insert(element: self.ElemT): Set { type ElemT = self.ElemT }
  ...
}

// Main program:
   let apple:     Fruit                      = ...
in let fruit_set: Set { type ElemT = Fruit } = ...
in fruit_set.insert(apple)
\end{lstlisting}


\caption{Immutable Set container type (Nominal Wyvern)}
\label{fig:nominal_wyvern_syntax_eg}
\end{figure}

We will now delve into certain noteworthy aspects of the Nominal Wyvern design.

\subsubsection{A Heterogeneous Typing Approach}
\label{sec:design:binary_typing}

The main difference between Nominal Wyvern and DOT is Wyvern's heterogeneous typing system. In DOT, nominal types are encoded using type members of some object, whereas, in Nominal Wyvern, there are two sorts of types: concrete object types (\emph{named} types) and abstract member types. At first glance Nominal Wyvern may seem like an extraction of all structural types (as well as any width-expanding refinements) to a global object type in DOT, but there is a more fundamental difference between named types and type members that contributes to Nominal Wyvern's goals of usability and decidability.

The key difference between named types and type members is the way they are specified. Type members are declared with a bound while named types are defined as named records that represent entities with the given properties. Semantically, named records more closely resemble the usual definition of a ``definition'': if we think of each Apple as a record with an integer ID and a floating point weight, then that is exactly what the named type ``Apple'' is. Contrast this with the declaration of a type member ``Apple'' with a bound on its subtype relation, which more closely resembles a guideline and guarantee on how generic/specific this type may be for future instantiators and users. This dichotomy separates responsibility by making named types the definer of types, and type members merely users of the pre-defined types. Consequently, all types that appear in a program (except the native top and bottom types) find their structures at the top-level.

The separation of named types from member types also warrants a separate way of defining the subtype relation between named types. Traditionally in DOT, type bounds perform two roles with different semantic meanings: one can use a type bound to either define a ``guarantee'' on a particular type member or define a subtype relation between nominal types. Nominal Wyvern separates the two semantically distinct roles with two separate constructs. Type bounds still exist but are used only to specify guarantees on type members.
Complementing type members are explicit subtype declarations, whose sole purpose is to define the nominal subtype relation between pre-defined named types. E.g.\ if we think of McIntosh as a special kind of apple, then we can explicitly declare this subtype relation between the two named types \code{McIntosh} and \code{Apple}. We can think of the explicit subtype declarations as defining base cases for the substitution principle. They are the only base cases because depth subtyping in the form of type refinement is automatic.


In addition to the usability benefits of semantic separation (e.g.\ clearer code), explicit subtype declarations also allow for more flexibility than traditional bounds in two ways:

\begin{enumerate}
  \item \emph{Multiple Subtyping}: Type T can be a subtype of multiple types. This is often used when one type wants to have the features of many other types. For example, a resource type (such as Apple) can declare itself as a subtype of both Equatable and Hashable so types that require either one (e.g.\ Set, Hashtable) can use it as the key.
  \item \emph{Conditional Subtyping}: 
  Nominal Wyvern supports a conditional subtyping syntax \code{n r <: n'}, which declares that the named type \code{n} is a subtype of \code{n'} whenever it is refined with \code{r}.
  The syntax allows a refinement on the LHS of the subtype symbol (<:) to make a subtype relation hold conditionally. This can be used due to either structural incompatibility or semantic incompatibility.

\end{enumerate}

This difference carries over to the subtyping algorithm. While the type bound on a type member S represents the authoritative ``next'' type to check after S, the subtype declarations with S as the base type on the LHS present us with multiple conditional options for what this type could also be seen as by the substitution principle. Section \ref{ch:decidability} details the heterogeneous subtyping algorithm and why it is a decidable problem after material/shape separation is applied.

To aid in nominal subtyping, we define a ``nominal subtyping graph'' to capture the multiple conditional subtyping relations between named types.

\begin{definition}[Nominal subtyping graph]
\label{defn:nominal_subtyping_graph}
For a set of top-level declarations $\overline{D}$, the nominal subtyping graph is a graph $\langle V, E \rangle$. The vertices, $V$, consist of all the named types in $\overline{D}$. The edges, $E$, each represent an explicit subtype declaration in $\overline{D}$, with the refinement labeled on the edge:
{\footnotesize
\begin{mathpar}
  \inferrule
  {n_1 \, r_1 \subtypes n_2 \in \overline{D}}
  {n_1 \xrightarrow{r_1} n_2 \in E}
\end{mathpar}}
\noindent
where $n_1 \xrightarrow{r_1} n_2$ is a directed edge from $n_1$ to $n_2$ labeled with $r_1$.
\end{definition}


   

\subsubsection{Top-Level Well-Formedness}

\begin{figure}[!t]
\input{wellformed-top}
\caption{Nominal Wyvern Top-Level Well-Formedness}
\label{fig:nominal_wyvern_ctx_wf}
\end{figure}

Figure \ref{fig:nominal_wyvern_ctx_wf} presents the judgment rules for top-level declaration well-formedness and type-checking.
The typing judgments for Nominal Wyvern are parameterised by up to four contexts: definitions $\Delta$, subtyping $\Sigma$, variable-typing $\Gamma$ and store-typing $S$.
The first two contexts, $\Delta$ and $\Sigma$, are static contexts which are derived solely from the top-level declarations $\overline{D}$, and are used in the typing and subtyping rules detailed in the later sections.
The variable-typing environment $\Gamma$ is standard.
The store-typing context $S$ maps from heap locations to types; it is required for the type safety proofs, but is empty when statically type-checking a program.

Judgment form $P : \tau$ denotes that a program $P$ is well-typed and results in a value of type $\tau$; the definitions given in the program are used to construct the initial typing context.
The rule checks that all declared subtyping relations are valid, and that the expression itself is well-typed.
Note that all named types are considered to be declared at the same time (in the same scope and can reference each other).

Judgment form $\Delta | \Sigma \vdash n_1 \, r_1 \subtypes n_2$ denotes that the type named $n_1$, when refined with refinement $r_1$, is a valid subtype of the type named $n_2$.
The rule unfolds the definition of the two named types and checks that all of the type members, fields and methods are structurally compatible.
The rule uses the subtyping rules for top-level declarations, which are defined in \autoref{sec:subtyping:top-level}.



\subsection{Material/Shape Separation}
\label{sec:design:ms_separation}
The definition of materials and shapes is based on the discovery by \citet{greenman:shapes} that a subset of types are responsible for creating cycles during subtype checking. These types are called shapes, because their role is not to serve as types for concrete objects, but rather to place recursive constraints on type parameters in the setting of Java (or type members, in our setting). In particular, shapes are used for the sole purpose of F-bounding some other type to guarantee that the other type has certain features (e.g. \code{Equatable} shows that whatever subtypes it has an \code{equals()} method).


As a consequence of being used in this way, the types annotated as shapes in a program must have a set of characteristics that separate them from materials, as described in our formal definition of Material/shape separation:

\begin{definition}[Material/shape separation]
\label{defn:ms_separation}
A program is properly material/shape separated if its types are annotated such that:
\begin{itemize}
  \itemsep0em
  \item A shape is never used as part of a lower bound (i.e.\ never appears after $\ge$ or $=$).
  \item The upper bound of a shape is always a shape, and named shapes can only subtype named shapes.
  \item No shapes are refined within refinements.
\end{itemize}
\end{definition}

From a semantic separation standpoint, the intuition behind the rules follows how shapes are supposed to be used in programs.
First of all, since shapes constrain types to have certain members (e.g. \code{equals()}), it only makes sense for them to serve as upper bounds of other types. This view of shapes also means that whatever type a shape subtypes must be another shape, otherwise the substitution principle would allow shapes to be used as materials.
Finally, a shape should not be used in refinement when there are no self-variables in scope to recursively bound a type.
The rules encourage the programmer to declare their F-bounded types as pre-declared structures instead of refining pre-declared structures to utilize F-boundedness later.
This definition is partly inherited from earlier works by \citet{greenman:shapes} and \citet{mackay:decidable_wyvern}, and partly created for Nominal Wyvern. Section \ref{sec:design:compare_decidable_wyvern} details our contribution over existing adaptations. 

While the intuitive use and characteristics of shapes are defined above, in order to achieve decidability, we need to actually make sure that shapes are the only way that cycles can appear during subtype checking. We do this by building a subtype dependency graph, following \citet{greenman:shapes}, that captures the way that subtyping involving some type T depends on subtyping relationships between other types T'. The only way that subtyping can fail to terminate, then, is if there are cycles in this graph. The insight from \citet{greenman:shapes} is that it is OK to allow cycles if each of them is guarded by a shape:

\begin{definition}[Shape type validity]
  After removing all edges labelled with at least one shape from the subtype dependency graph, the graph is acyclic.
\end{definition}

The intuition behind the subtype dependency graph is that if $\beta$ has an edge to (``depends on'') $\beta'$, then when $\beta$ is encountered as a base type during subtype derivation, $\beta'$ may appear as the base type later in the derivation. The graph is thus a conservative guarantee of what types will not be visited again in later derivations. This serves as the key idea for proving decidability in Section \ref{ch:decidability}.

Now we will formally define the subtype dependency graph.

\begin{definition}[Subtype dependency graph (SDG)]
\label{def:sdg}
For a set of top level declarations $\overline{D}$, the subtype dependency graph is a graph $\langle V, E \rangle$. The vertices, $V$, consist of all available base types, which are: $\top$, $\bot$, all declared named types in $\overline{D}$ and their type members. To disambiguate type members with the same name from different named types, each type member $t$ of name $n$ is denoted $n \dcolon t$ (which will be referred to as a ``pseudotype'').

The edges, $E$, are generated as follows. Define recursive metafunction {\small \texttt{\textup{GenEdges}}} as:\\
${\small \texttt{\textup{GenEdges}}(\beta, \{ \cdot \}, \beta_r, \overline{\beta_a}) = \beta_r \xrightarrow{\overline{\beta_a}} \beta}$ (GE-EMP)\\
${\small \texttt{\textup{GenEdges}}(\beta, \{ \overline{\delta'}, \delta \}, \beta_r, \overline{\beta_a}) =}$\\
${\small ~~~ \texttt{\textup{GenEdges}}(\beta, \{ \overline{\delta'} \}, \beta_r, \overline{\beta_a}),}$\\
${\small ~~~ \texttt{\textup{GenEdges}}(\beta_t, r_t, \beta_r, \overline{\beta_a'})}$ (GE-LIST)\\
where $\delta = \Type{t}{B}{\beta_t \, r_t}$ and $\overline{\beta_a'} = \overline{\beta_a}, \beta$.\\
\\
\noindent For each type member declaration, {\small $\Type{t}{B}{\beta_t \, r_t}$}, in named type $n$, generate edges with\\ $\texttt{\textup{GenEdges}}(\beta_t , r_t, n \dcolon t, \cdot)$. Then, for each subtype declaration {\small $\Subtype{n_1}{r_1}{n_2}$} in $\overline{D}$, generate edges:
{\small
\begin{mathpar}
  \inferrule*[Right=back]
  {\\}
  {n_2 \xrightarrow{\cdot} n_1}
\and
  \inferrule*[Right=back-ref-root]
  {\Type{t_r}{\_}{n_r \, r_r} \in^* r_1
  }
  {n_2 \xrightarrow{\cdot} n_r}
\end{mathpar}
}
{\noindent where $\delta \in^* r$ is true if the refinement member $\delta$ appears anywhere syntactically within $r$ (can be nested inside any number of refinements), and $\delta \in r$ is true only if $\delta$ is immediately within $r$ (not further nested), and
where $v_1 \xrightarrow{\cdot} v_2$ represents an unlabeled edge, and $v_1 \xrightarrow{\beta_1 \beta_2 \ldots} v_2$ represents an edge labeled with the base types $\beta_1, \beta_2, \ldots$ in that order.
}
\end{definition}

\begin{example}
\label{ex:subtype_dep_graph}
Recall the Fruit/Set example from Figure \ref{fig:nominal_wyvern_syntax_eg}, which had the following types and type members:

\begin{lstlisting}[mathescape, language=Scala, style=nwyvern, escapechar=|, basicstyle={\scriptsize \ttfamily}]
@shape name Equatable { self =>
  type EqT >=              $\bot$
}

name Fruit { self =>
  type EqT =               Fruit
}

subtype Fruit <: Equatable

name Set { self =>
  type ElemT <= Equatable { type EqT >= self.ElemT }
}
\end{lstlisting}
The subtype dependency graph for the above set of type definitions is:

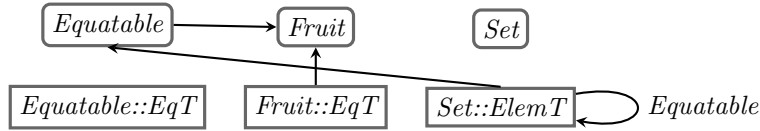
\begin{center}
\begin{tikzpicture}[roundedrectnode/.style={squarednode, rounded corners=1mm}, node distance=5mm]
  \node[squarednode]  (eq_t)                                 {Equatable\dcolon EqT};
  \node[squarednode]  (fruit_t)    [right=of eq_t]           {Fruit\dcolon EqT};
  \node[squarednode]  (set_elemt)  [right=of fruit_t]        {Set\dcolon ElemT};
  
  \node[roundedrectnode]  (eq)     [above=of eq_t]           {Equatable};
  \node[roundedrectnode]  (fruit)  [above=of fruit_t]        {Fruit};
  \node[roundedrectnode]  (set)    [above=of set_elemt]      {Set};

  \draw[arrow] (fruit_t.north) -- (fruit.south);
  \draw[arrow] (eq.east) -- (fruit.west);
  \draw[arrow] (set_elemt.north) -- (eq.south);
  \draw[arrow] (set_elemt) to [in=-10, out=10, loop] node [right] {Equatable} (set_elemt);
\end{tikzpicture}
\end{center}

The \code{Set::ElemT} pseudotype depends on itself and \code{Equatable}; it transitively depends on \code{Fruit}, as \code{Fruit} is declared as a subtype of \code{Equatable}.
Intuitively, if we want to check that a type is a subtype of \code{Set::ElemT}, then we may have to check against \code{Set:ElemT}, \code{Equatable} and \code{Fruit}.
Similarly, if we want to check that a type is a \code{Fruit::EqT}, then we may need to look at the definition of \code{Fruit}.
Note that, although \code{Equatable} and \code{Equatable::EqT} are lexically related, we have no dependency between them, as they are not related by subtyping.

\end{example}


\subsubsection{Comparison with Decidable Wyvern}
\label{sec:design:compare_decidable_wyvern}

Both Nominal Wyvern and Decidable Wyvern\footnote{Since both \emph{Wyv}$_{\textit{self}}$ and \emph{Wyv}$_{\textit{fix}}$ of \citet{mackay:popl2020} use the same material/shape separation rules as Decidable Wyvern \cite{mackay:decidable_wyvern}, we consider them together in the following comparison.} achieve decidable subtyping with material/shape separation. The novelty of Nominal Wyvern is our insight that adding nominality to a purely structural system such as DOT makes the material/shape separation rules simpler. Material/shape separation, as originally proposed, relies on identifying type parameters in Java, and we can do the same in Nominal Wyvern trivially since type refinements take on a similar role. But Decidable Wyvern, similar to DOT, does not have the semantic separation of ``type definition'' and ``type refinement'' that is present in both Nominal Wyvern and Java. Instead, both use the same syntax, making it hard to figure out when a type is being refined and when one is being defined just by looking at the code. As a result, Decidable Wyvern had to add more rules in order to artificially limit the kind of refinements that are allowed so that any valid system ends up with a type hierarchy similar to that of a nominal system, from which it can then infer the purpose of each type refinement. These rules make the Decidable Wyvern system more complex for language designers, tools developers, and programmers, and it isn't obvious how to hide this complexity from end users of the language.

One clear benefit of Nominal Wyvern's design (i.e. semantically separating type definitions from refinements) is that Nominal Wyvern can easily disallow self-references in type refinements, while this would be impossible in Decidable Wyvern without somehow breaking recursive types. Not only does this simplification help with our theme of furthering usability/familiarity with semantic separation (in this case, making clear that refinements are to be used for refining existing bounds, instead of setting up new connections between type members that are previously unintended), it also brings us closer to decidability by limiting the kind of cycles that can occur. As explained in this paper, we can also provide a clear intuition for what ``shape'' types are in Nominal Wyvern, whereas the concept of a ``shape'' is more abstract (more reliant on complex inference rules), and thus likely harder to understand and apply, in Decidable Wyvern.

This difference is reflected in how the separation rules are defined. For Nominal Wyvern, the separation rules are treated as an additional requirement that needs to be separately enforced during type checking. In contrast, Decidable Wyvern's separation rules are tightly integrated with its grammar. The benefit of having a syntactic restriction means the separation rules can be easily specified along with their grammar and can be easily checked by existing parsers. However, this also means the grammar involves many nuanced details that may be hard to follow: Users will need to always be aware of whether they are currently programming a shape or material, and be wary of what syntactic constructs are not allowed. Nominal Wyvern trades off the easily specifiable separation rules in favour of a uniform grammar (i.e. material/shape agnostic) with an additional separation check.

On a more technical level, the dependency graph construction process is, however, simpler in Decidable Wyvern because it is based on a more uniform type system (instead of the heterogeneous typing approach used here). In contrast, Nominal Wyvern's subtype dependency graph is separated into nodes that are named types and nodes that are type members (i.e. pseudotypes). By sacrificing uniformity, however, Nominal Wyvern allows the subtype dependency graph to be easily partitioned into $N+1$ separate graphs that can be individually checked for cycles (where $N$ is the number of named types in the top-level): a sub-graph for each named type $n$ that only contains the edges between the pseudotypes of $n$, and another sub-graph with all the edges between named types. Since a named type node only points to named type nodes, and a pseudotype node $n \dcolon t$ only points to named types and other sibling pseudotypes of $n$, there will be no cycles that span multiple sub-graphs. This separation makes it easier to check graphs for cycles, and also prevents long-spanning dependencies that may be hard to understand. This reduced complexity also means developers can now more easily identify potential problems in their types.

\subsection{Term Typing}
\label{sec:design:term_typing}

\begin{figure}[tb!]
  \input{typing}
  \caption{Nominal Wyvern Term Typing}
  \label{fig:nominal_wyvern_term_typing}
\end{figure}

In this section, we introduce the typing rules for Nominal Wyvern terms.

All expressions in Nominal Wyvern correspond to objects.
All expressions are given names, either assigned with a \code{let} expression or as a \code{val} member in another object.
Objects are then used by referring to the path that refers to their names, either directly as the assigned variable in a \code{let} expression, or by selecting a \code{val} member from another object.
This restricted form simplifies the metatheory; it is analogous to A-normal form.

The judgments for typing expressions in Nominal Wyvern are shown in Figure \ref{fig:nominal_wyvern_term_typing}.
The majority of our judgments are parameterised by four contexts: $\Delta$, $\Sigma$, $\Gamma$ and $S$.
The most relevant contexts for typing are the variable-typing context $\Gamma$ and the location-typing context $S$. The context $\Gamma$ keeps track of the type of each variable as usual.
The context $S$ keeps track of the concrete type of objects for evaluation; it is used in the type safety proof in Section \ref{sec:proof}.
The other two contexts, $\Delta$ and $\Sigma$ are top-level contexts that contain the type definitions and subtyping relations respectively.

We introduce a specific judgment form for typing singleton paths $\JGSType{p}{\tau}$; the reason for this separation is to avoid mutual recursion between the later definition of subtyping (Section \ref{sec:decidability:judgments}) and the full expression-typing judgment $\JGSType{e}{\tau}$.
The path-typing judgment contains rules {\sc T-Var} for variables and {\sc T-Loc} for locations, both of which are standard.

In the full expression-typing judgment, rule {\sc T-Sel} is used for type-checking field selection $p.v$, which involves first typing the singleton path $p$ to type $\tau$.
The type $\tau$ itself may be a path-dependent type, so we use the \emph{exposure} operation $\JGSExpose{\tau}{n \, r}$ to unfold it to a named type.
We then use the type declaration lookup judgment $\JGSMember{n \, r}{p}{\Val{v}{\tau_v}}$ to find the type of the field in the named type definition in $\Delta$ and apply any refinements from $r$.
The type of the field $\tau_v$ may refer to the type definition's self-variable, so the type declaration lookup judgment takes the target object's path $p$ and substitutes any references to the self-variable for $p$ in the field type.
The type declaration lookup and exposure judgments are defined in Section \ref{sec:design:exposure}; exposure is analogous to the exposure operation from \citet{nieto:algorithmic_typing_dot} and \citet{hu:decidable_dsub_fragments}.

Rule {\sc T-Let} for let-expressions uses the \emph{avoiding exposure} operation $\Avoid{\tau}{\le}{x}{\tau'}$, which ensures that the binding $x$ is not mentioned in the result type $\tau$.
The avoiding exposure fails if the type cannot be unfolded to avoid mentioning $x$, or if doing so would require more than some arbitrary number of unfoldings.
The $\le$ superscript here denotes that avoiding can \emph{overapproximate} the original type by upcasting to a supertype: that is, $\tau \le \tau'$.
This avoidance is required as Nominal Wyvern supports a limited form of dependent types, and the types of expressions may refer to the let-bindings.
This operation is defined in Subsection~\ref{sec:design:avoid}.

Rule {\sc T-App} is used for type-checking a method application $p.f(p')$.
We first check that the method target $p$ has type $\tau$, which we expose to a named definition $n \, r$; we then find the method named $f$.
Next, we check that the argument type $\tau'$ is a subtype of the parameter type $\tau_a$.
Finally, we return the method's result type $\tau_r$ and substitute the method parameter $x_a$ with its argument $p_a$.

Rule {\sc T-New} is used for type-checking object allocation $\New{\tau}{x}{\overline{d}}$; the type of the allocated object $\tau$ is specified explicitly.
Requiring explicit types is in line with the goals of a nominal type system. Without a name provided by the programmer, a structure can potentially be mapped to many unrelated (w.r.t. subtyping) named types. However, the explicit type given to \code{new} need not be exactly the same type as the following structure: it is only required that the structure is a proper structural subtype to the exposed type. This differing view of the same object provides an easy way to abstract the types and other members of the newly-allocated object.
We first check that the specified type $\tau$ is a \emph{valid} type with the judgment form $\JGSWf{\tau}$, and then check that the object definition is well-typed with the judgment form $\JGSType{\{ x \Rightarrow \overline{d} \}}{\tau}$.

The judgment form $\JGSType{\{ x \Rightarrow \overline{d} \}}{\tau}$ checks that an object definition is well-typed.
We define $\Gamma'$ by extending the context with the type of the newly-allocated object $\tau_x$, defined by merging the specified type $\tau$ with any refinements from $d$.
We check that $\tau_x$ satisfies the original type's requirements.
Then, we check that all fields and methods are well-typed in the extended environment $\Gamma'$.


Finally, the judgment form $\JGSWf{\tau}$ ensures that the type signature of its members forms a structural type that is a subtype of the declared type of the object.
Validity for types $\top$ and $\bot$ are trivial.
To check validity for a refined type $\beta~r_\beta$, we intuitively want to check that all of the refinements in $r_\beta$ are valid.
First, we expose the \emph{unrefined} base type $\beta$ to a type $\tau_u$.
For each type member in the refinement $r_\beta$, we check that there is a corresponding type member in the exposed unrefined type, and check that refined type member is a subtype of the unrefined member.
We extend the context with the self-binding $x$, giving it the overall refined type $\beta~r_\beta$.

\subsection{Normalising Types}
\label{sec:design:exposure}

\begin{figure}
\input{lookup}
\caption{Declaration Lookup}
\label{fig:nominal_wyvern_lookup}
\end{figure}

In this section, we define the auxiliary judgment forms that are used by the typing and subtyping rules.
\autoref{fig:nominal_wyvern_lookup} defines the declaration lookup judgment form $\JGSMember{\tau}{p}{\sigma}$.
This judgment form denotes that the type $\tau$ has a declaration $\sigma$.
For path-dependent types, the declarations may refer to the current object or self variable as $p$.
The judgment expects that $p$ has type $\tau$ or some subtype of $\tau$, but we do not explicitly require this in the rules as it would entail a circular definition.
Intuitively, an algorithmic version of this judgment would take the declaration's name as input, as well as the path $p$ and the contexts, and return the declaration itself.

Rule {\sc Look-Refine} applies when looking for a type member $t$ inside a refinement.
We look up the refinement and return it as-is.

Rule {\sc Look-Name} applies when looking for a type member which has not been overwritten by a refinement, as well as looking for methods and fields (which cannot be refined).
In this case, we look up the type definition in the type definition context $\Delta$, and find the corresponding declaration $\sigma$.
Type definitions may refer to the self variable $x_n$, so we substitute it with the current object's path $p$.

\begin{figure}
\input{exposure}
\caption{Exposure, Upcasting and Downcasting}
\label{fig:nominal_wyvern_exposure}
\end{figure}

To support path-dependent types, we need a limited form of normalization for types.
\autoref{fig:nominal_wyvern_exposure} defines the judgment forms for \emph{exposure}, \emph{upcasting}, and \emph{downcasting}.
These operations are analogous to the exposure and casting operations from \citet{hu:decidable_dsub_fragments} and \citet{nieto:algorithmic_typing_dot}.

Exposure normalizes a type to a supertype.
Previous presentations also ensure that the exposed type is not path-dependent; in our system, references to type members defined by a lower-bounds do not have a supertype other than $\top$, so we define exposure to leave them as-is.
Intuitively, exposure works by looking up any path-dependent types to find their upper-bound, and recursively exposing the result.
Upcasting works similarly and finds the upper-bound of path-dependent types, but does not recursively simplify the result.
Downcasting is the dual operation, and finds the lower-bound of path-dependent types; like upcast, it does not recursively simplify the result.

Judgment form $\JGSExpose{\tau}{\tau'}$ denotes that $\tau$ exposes to $\tau'$.
Rules {\sc Exp-Top}, {\sc Exp-Bot}, and {\sc Exp-Name} leave their types as-is, as they are already in a head-normal-form.

Rule {\sc Exp-Upper} applies when looking up a type member $p.t$ that is defined with an upper-bound ($\le$) or an exact-bound ($=$).
First, we check that the path $p$ has some type $\tau_p$; note that the judgment $\JGSType{p}{\tau_p}$ uses the simpler path-specific judgment form, rather than requiring the full expression typing judgment.
We expose $\tau_p$ to a normalized type $\tau_p'$ and look for the type member named $t$.
To get the result type, we first instantiate the member's self-type $x$ with the actual path $p$, and apply any refinements $r$ that were specified on the original path-dependent type ($p.t~r_1$).
Finally, we normalize the instantiated refined type to $\tau'$, which is the overall result.

Rule {\sc Exp-Otherwise} applies when looking up a type member $p.t$ that is not defined with an upper-bound.
In this case, the type member may be defined with a lower-bound, or the type of $p$ may not expose to a concrete upper bound.
Exposure may result in a supertype of the original type, but not a subtype; we therefore cannot use the lower-bound type member.
Instead, we return the type as-is.

Judgment form $\JGSUpcast{\tau}{\tau'}$ denotes that $\tau$ can be \emph{upcast} to $\tau'$.
Upcasting is very similar to exposure:
in a sense, upcast only performs a single step of the unfolding of $p.t$, while exposure performs as many steps as necessary.
Rule {\sc Uc-Upper} applies when upcasting an upper-bound type member, and is similar to a non-recursive version of {\sc Exp-Upper}.
Although the definition of upcast is not recursive, it does use exposure to expose the type of $p$.
For all other types, including paths with lower-bounds, rule {\sc Uc-Otherwise} returns the type as-is.

Upcasting is used in the subtyping rules (\autoref{ch:decidability}); it gives the subtyping fine-grained control over exactly \emph{how much} to expose.
For example, a type $p.t$ in a suitable environment may expose to $\top$, but it might require multiple upcast steps, going via an intermediate path $q.t'$. If we wish to check whether $p.t$ is a subtype of $q.t'$, we cannot expose $p.t$ all the way to $\top$.
(We also cannot expose $q.t'$ to $\top$, as checking whether $p.t$ is a subtype of $\top$ is a semantically different question.)

Judgment form $\JGSDowncast{\tau}{\tau'}$ denotes that $\tau$ can be \emph{downcast} to $\tau'$.
Downcasting is the dual to upcasting, and involves resolving path-dependent types through lower-bound type members.
Rule {\sc Dc-Lower} applies when downcasting a lower-bound type member.
For all other types, rule {\sc Dc-Otherwise} returns the type as-is.

\subsubsection{Avoidance}
\label{sec:design:avoid}

As Nominal Wyvern is dependently typed, types can refer to local bindings.
This presents an issue when typing let-bindings, as the overall type of the let-binding cannot refer to the local binding itself.
To resolve this issue, we introduce the \emph{avoidance} judgment, which unfolds the type so that it does not mention the local let-binding.

\begin{figure}
\input{bound-ops}
\caption{Join ($\BoundJoin{B}{B}$) and product ($\BoundMul{B}{B}$) for bounds; join conceptually computes the union of two bounds, while product describes the result of applying the bound for a refinement inside the bound for the overall object.}
\label{fig:bound-ops}
\end{figure}

Figure~\ref{fig:bound-ops} defines some prerequisite operations on type bounds.
The join operation on bounds is partial and returns the union of the input bounds, if there is one.
The product operation is total and treats equality ($=$) as the zero element, subtyping ($\le$) as identity, and supertype ($\ge$) as the inverse.

\begin{figure}
\input{avoid}
\caption{Avoidance}
\label{fig:avoid}
\end{figure}

Figure~\ref{fig:avoid} defines the judgment form for avoidance.
Judgment form $\Avoid{\tau}{B}{x}{\tau'}$ denotes that type $\tau$ unfolds to $\tau'$ without mentioning binding $x$.
The bound $B$ denotes whether the unfolded type is equal to the original type ($\tau = \tau'$), a subtype ($\tau \le \tau'$) or a supertype ($\tau \ge \tau'$).
The judgment form for base types is similar; however, as we unfold path members, the result of avoiding a base type must be a full type.

Avoiding a non-path base type is straightforward as it requires no unfolding.
Rules {\sc Avoid-Top}, {\sc Avoid-Bot} and {\sc Avoid-Name} return their input type as-is, and return the equal bound to signify that the result is equivalent.

To unfold a top, bottom or a name in a context that requires a supertype or subtype rather than equivalence, the rule {\sc Avoid-Eq} applies.
The rule here states that, if $\beta$ unfolds to $\tau'$ with the bound restricted to equality, then we can also unfold $\beta$ to $\tau'$ for a subtype or supertype bound $B$.

For paths $p.t$, we check if the type is the one we are avoiding.
If not, rule {\sc Avoid-Path-NE} applies and leaves the type as-is.

For the actual case of avoidance, rule {\sc Avoid-Path-Eq} applies.
Here, we look up the type of $x$ and expose its type; then, we find the corresponding type member and note its bound $B$.
We then recursively avoid the type member's type, which gives a different bound $B'$.
Finally, we require that the type member's bound $B$ and the recursive avoid's bound $B'$ are compatible by joining them together in the result bound.

To avoid full types $\beta \, r$, rule {\sc Avoid-Type} applies.
Here, we unfold the base type and refinements separately.
For each refinement, we restrict the unfolding depending on both the type member's bound ($B_r$), as well as the overall avoid's bound ($B$), by unfolding the type member's type with respect to the product of the two bounds ($\BoundMul{B_r}{B}$).
This restriction ensures that, for example, if the caller requires unfolding with-respect-to-equality ($B = (=)$), then all type members are also unfolded with-respect-to-equality.

It is not always possible to unfold a binding to avoid a given binding.
Unfolding a type to avoid a binding does not always result in a finite type.
For example, consider the following \code{Loop} type:

\begin{lstlisting}[mathescape, language=Scala, style=nwyvern, escapechar=|, basicstyle={\scriptsize \ttfamily}]
  @shape name Comparable { self =>
    type T <= Top
    def compare(self.T): Int
  }

  name Loop { self =>
    type T <= Comparable { T <= self.T }
  }
\end{lstlisting}

Here, if we have a binding \code{x: Loop}, we might wish to unfold the binding's type member (\code{x.T}).
Unfolding the type one step results in \code{Comparable \{ T <= x.T \}}, which still mentions the binding $x$.
If we continued this unfolding, we would have an infinitely-nested sequence of \code{Loop}s.
We resolve this issue by instrumenting avoid with a \emph{fuel}, or an iteration limit.

We present a simplified avoidance here without the iteration limit; applying such a limit is standard, and the details of the limit are not particularly important.
With this iteration limit in place, we can assume that the fuel-instrumented avoidance is decidable.


\subsection{Exposure is Decidable}


\begin{definition}[Rank of variables, paths and types]
  \label{def:rank}
  We define the position of a variable $x$ in $\Gamma$ as its \emph{rank}: $\Rank{\Gamma}{x}$; the first (leftmost) variable has rank one.
  Rank is extended to paths and types by taking the maximum rank of each of its free variables, or zero if there are no free variables.
\end{definition}

\begin{definition}[Head-rank of types]
  \label{def:head-rank}
  The head-rank of a type is the rank of the outermost constructor of a type.
  The head-rank of a named type is zero, as it the head-rank of top and bottom types.
  The head-rank of a path-dependent type $p.t \, r$ is the rank of $p$ with respect to some environment $\Gamma$.
\end{definition}

\begin{definition}[Well-formedness of environments and types]
  \label{def:rank_valid}
  An environment $\Gamma$ is well-formed if each record $x: \tau$ in $\Gamma$ only mentions the variables that precede it in $\Gamma$:
  $$\forall i, j \in \Gamma.\ \Rank{\Gamma}{i} < \Rank{\Gamma}{j} \implies \Rank{\Gamma}{\Gamma(i)} < \Rank{\Gamma}{\Gamma(j)}$$
  Similarly, a type $\tau$ is well-formed with-respect-to $\Gamma$ if it only mentions variables in $\Gamma$.
\end{definition}

Intuitively, a well-formed environment means that the first variable can only refer to static types (names, top, or bottom), and higher-ranked variables can refer to path-dependent types whose paths are rooted at the lower-ranked variables (the variable that begins a path is denoted as the root of the path, and the rank of a path is defined as the rank of its root).
A well-formed environment $\Gamma$ can be extended to well-formed environment $\Gamma, x: \tau$ if $\tau$ only mentions variables in $\Gamma$, but $\tau$ cannot mention $x$.
Note that this definition of well-formedness does not include the typing \emph{validity} judgment $\JGSWf{\tau}$: typing validity uses the subtyping relation, and so we cannot assume it holds for our proofs of decidability of subtyping.


Path-typing and declaration-lookup are both non-recursive and straightforwardly decidable.

To prove that exposure is decidable, we must define a measure and show that each nested occurrence of exposure strictly decreases the measure.
We first introduce the informal argument before making it concrete.
Most of the rules in exposure are fairly straightforward to prove decidable: rules {\sc Exp-Top, Exp-Bot, Exp-Name} and {\sc Exp-Otherwise} do not have nested occurrences at all.

Rule {\sc Exp-Upper} has two nested occurrences of exposure.
Rule {\sc Exp-Upper} applies when exposing a path-dependent type $p.t~r$; it first finds the binding $p: \tau_p$, and exposes $\tau_p$ to $\tau_p'$.
In this case, the head-rank of $\tau_p$ is strictly smaller than the head-rank of the type $p.t$, as the binding $p: \tau_p$ cannot refer to $p$.
For this nested occurrence, it would be sufficient to use the head-rank of the input type as the decreasing measure.

After exposing the type $\tau_p'$, the rule uses the judgment $\tau_p' \ni_p \Type{t}{\le}{\tau_t}$ to find the type of the type member $p.t$, which has type $\tau_t$.
The rule then exposes $\tau_t$ to $\tau_t'$.
This second nested occurrence of exposure doesn't necessarily decrease the head-rank: the type $\tau_t$ \emph{can} refer to $p$, if the type member declaration refers to other type members defined in $\tau_p'$.

Instead, this second occurrence requires a different measure: if the type member $\tau_t$ is a path-dependent type $p.t' r'$, then the type members must refer to some pseudotype \code{n::t'} that precedes the original type \code{n::t} in the acyclic subtype dependency graph (definition~\ref{def:sdg}).
Otherwise, if the type member is not a path-dependent type on $p$, then the head-rank must be decreasing, as type $\tau_p$ can only refer to variables that precede $p$ in the environment.

Unfortunately, this informal argument is not straightforward to state formally: for a given path-dependent type $p.t$, finding the pseudotype \code{n::t} requires looking up the type of $p$ in the environment and exposing it.
Using exposure in our measure here would introduce a cyclic dependency in our proof of decidability.

Using the above informal argument as a guide, we will restructure the definition of exposure to make the measure easier to state.

\begin{figure}[t]
\begin{tabbing}
  M \= let \= MMMMMM \= \kill
  exposes: $\Gamma \to \Gamma$ \\
  exposes($\cdot$) \>\>\> $= \cdot$ \\
  exposes($\Gamma, x: \tau$) \>\>\> $=$ \\
    \> let \> $\Gamma'$ = exposes($\Gamma$) \\
    \> in  \> $\Gamma', x:$ expose1($\Gamma', \tau$) \\
  \\
  expose1: $\Gamma \times \tau \to \tau$ \\
  expose1($\Gamma, \top$) \>\>\> $= \top$ \\
  expose1($\Gamma, \bot$) \>\>\> $= \bot$ \\
  expose1($\Gamma, n~r$) \>\>\> $= n~r$ \\
  expose1($\Gamma, x.t~r$) \\
    \> $|~ x: \tau \in \Gamma, ~ \tau \ni_x \Type{t}{\substack{\leq \\ =}}{\tau'}$ \\
    \> \> $ = $ expose1($\Gamma, \tau'$) $+_r r$ \\
    \> $|~ $ otherwise \\
    \> \> $ = x.t~r$ \\
\end{tabbing}
\caption{Algorithmic terminating exposure}
\label{fig:exposure_algorithmic}
\end{figure}

\autoref{fig:exposure_algorithmic} defines functions \emph{exposes} and \emph{expose1}. Function \emph{exposes} takes an environment and applies \emph{expose1} to each type, building up a new environment.
Function \emph{expose1} takes an already-exposed environment and a type, and exposes the type.
These functions correspond to an eager version of the exposure rules.

The function \emph{exposes} is trivially terminating, as the length of the environment $\Gamma$ reduces with each recursive call.

\begin{theorem}[Single-exposure \emph{expose1} is decidable]
\end{theorem}
\begin{proof}

To prove termination of \emph{expose1}, we define a lexicographic measure \emph{expose1-measure}:

\begin{tabbing}
  M \= let \= MMMMMM \= \kill
  expose1-measure: $\Gamma \times \tau \to (\mathbb{N} \times \mathbb{N})$ \\
  expose1-measure($\Gamma, \tau$) $=$ \\
    \> $(\RankHead{\Gamma}{\tau}, \mbox{subtype-dependency-measure}(\Gamma, \tau))$ \\
  \\
  subtype-dependency-measure: $\Gamma \times \tau \to \mathbb{N}$ \\
  subtype-dependency-measure($(\Gamma, x: n~r, \Gamma'), x.t~r$) $=$ \\
      \> index(subtype-dependency-graph, n::t) \\
    subtype-dependency-measure($\Gamma, \tau$) $= 0$ \\
\end{tabbing}

This measure states that on each recursive call, either the rank of the head of the type $\tau$ decreases, or if the rank remains the same, then $\tau$ must refer to some pseudotype \code{n::t}, and the index of the pseudotype in the topological ordering of the subtype dependency graph decreases.

We show that for a well-formed context $\Gamma$ and binding $x: \tau \in \Gamma$, with type member $\tau \ni_x \Type{t}{\le}{\tau'}$, then
 $\mbox{expose1-measure}(\Gamma, \tau') < \mbox{expose1-measure}(\Gamma, \tau)$.
 First, we note some facts:

By well-formed assumption, have $\Rank{\Gamma}{\tau} < \Rank{\Gamma}{x}$.

Note that $\Rank{\Gamma}{\tau'} \le \Rank{\Gamma}{x}$:
declaration-lookup uses types from the refinement, or the top-level definition with the self-variable substituted with $x$. For a refinement, the rank will be included in the maximum of $\Rank{\Gamma}{\tau}$, so the rank of $\tau'$ will be at most that of $\tau$. For a top-level definition is closed except for the self-variable which, after substitution, will have rank at most $x$.

We start by case analysis of $\tau'$:
\begin{description}
  \item[Case] $\top$, $\bot$ and $n~r$: \\
    Direct: the head-rank of $\tau'$ is zero, while the head-rank of $\tau$ is positive (we are in a recursive call to \emph{expose1}, so $\tau$ is of the form $x.t$).
  \item[Case] path-dependent type $q.t'$: \\
    We first check if $x = q$:
    \begin{description}
    \item[Subcase] $x \not= q$: \\
      Direct: head-rank of $q$ is strictly less-than $x$, as maximum rank of $\tau'$ is $x$.
    \item[Subcase] $x = q$: \\
      By inversion of the declaration-lookup, we have two cases:
      \begin{description}
        \item[Inversion] {\sc Look-Refine}:\\
          In this case, the type member definition $\Type{t}{B}{x.t'}$ exists in the refinement of $\tau$. \\
          However, the rank of $\tau$ is strictly less than the rank of $x$, and so this case is impossible.
        \item[Inversion] {\sc Look-Name}:\\
          The type member $\Type{t}{B}{x.t'}$ exists in the top-level type definition. \\
          Thus, $\tau'$ must be a named type $n \, r$. \\
          The pseudotype \code{n::t} thus has an edge to \code{n::t'} in the subtype dependency graph. \\
          Thus, the head-rank remains the same in the measure, while the topological index of \code{n::t'} is strictly less-than the topological index of \code{n::t}.
      \end{description}
  \end{description}
\end{description}

\end{proof}

\begin{theorem}[Single-exposure weakening]
  Given a well-formed environment $\Gamma, \Gamma'$ and type $\tau$ well-formed with-respect to $\Gamma$, $\mbox{expose1}(\mbox{expose}(\Gamma), \tau) = \mbox{expose1}(\mbox{expose}(\Gamma, \Gamma'), \tau)$.
\end{theorem}
\begin{proof}
  By induction on $\Gamma'$.
\end{proof}

\begin{theorem}[Single-exposure implies exposure judgment]
  For well-formed environment $\Gamma$ and type $\tau$, $\JGSExpose{\tau}{\mbox{expose1}(\mbox{exposes}(\Gamma), \tau)}$.
\end{theorem}
\begin{proof}
  By strong induction on the measure $\mbox{expose1-measure}(\mbox{exposes}(\Gamma, S), \tau)$.
  \\
  Case analysis on $\tau$:
  \begin{description}
    \item[Case] Top, Bottom and Name: \\
      Apply rules {\sc Exp-Top, Exp-Bot} and {\sc Exp-Name} as necessary.
    \item[Case] $x.t~r$:\\
      By well-formed assumption, $x: \tau_x \in \Gamma$. \\
      Let $\tau_x' = \mbox{expose1}(\mbox{exposes}(\Gamma), \tau_x)$. \\
      By inductive hypothesis, with decreasing $\RankHead{\Gamma}{\tau_x} < \RankHead{\Gamma}{x.t~r}$, we have
        $\JGSExpose{\tau_x}{\tau_x'}$. \\
      (The remainder follows the recursive structure of the $\mbox{expose1}$ function.) \\
      By lemma single-exposure-weakening, $x: \tau_x' \in \mbox{exposes}(\Gamma)$.\\
      Case analysis on type member lookup of $\tau_x'$:
      \begin{description}
        \item[Case] $\tau_x' \ni_x \Type{t}{\le}{\tau'}$; or $\tau_x' \ni_x \Type{t}{=}{\tau'}$: \\
          Let $\tau'' = \mbox{expose1}(\mbox{exposes}(\Gamma), \tau_x)$. \\
          By inductive hypothesis, we have $\JGSExpose{\tau'}{\tau''}$. \\
          Have $\mbox{expose1}(\mbox{exposes}(\Gamma), \tau) = \tau'' +_r r$. \\
          Apply rule {\sc Exp-Lower}.
        \item[Case] $\tau_x' \ni_x \Type{t}{\ge}{\tau'}$; or type member lookup fails: \\
          Have $\mbox{expose1}(\mbox{exposes}(\Gamma), \tau) = \tau$. \\
          Apply rule {\sc Exp-Otherwise}.
      \end{description}
  \end{description}
\end{proof}

The above lemma ensures that, given a well-formed context and type, we can always construct an exposure derivation tree that matches the expose1 function.
To finish the proof of decidability, we could prove the converse direction; however, it is sufficient to prove that the exposure judgment is deterministic: for a given input environment and type, there is at most one unique result type.

\begin{theorem}[Exposure is decidable]
  For a well-formed environment $\Gamma$ and a correspondingly well-formed type $\tau$, exposure $\JGSExpose{\tau}{\tau'}$ is decidable.
\end{theorem}
\label{lma:exposure-is-decidable}

\begin{proof}
  By single-exposure implies judgment and determinism of exposure.
\end{proof}

\section{Subtyping}
\label{ch:decidability}
\label{sec:decidability:judgments}

In this section, we present the basic subtyping rules and decidability proof (\autoref{sec:decidability:decidability}).
After proving decidability of the base system, we describe an extension that allows for more expressivity, while retaining decidability (\autoref{sec:increasing-expressivity:extension}).

\begin{figure}
\input{04-subtyping}
\caption{Nominal Wyvern Subtyping}
\label{fig:nominal_wyvern_subtyping}
\end{figure}

Figure \ref{fig:nominal_wyvern_subtyping} presents the subtyping judgments for Nominal Wyvern types.
These rules consider only types; we postpone subtyping for member declarations to \autoref{sec:subtyping:top-level} as their decidability argument is relatively straightforward.
The subtyping judgments for types consist of three main judgment forms:

\begin{enumerate}
  \item \emph{Nominal and path-dependent subtyping} $\JGSSubtype{\tau}{\tau'}$: subtyping between base types, named types and path-dependent types.
  \item \emph{Type member subtyping} $\JGSSubtype{\Type{t}{B}{\tau}}{\Type{t}{B}{\tau'}}$: subtyping between two type members.
  \item \emph{Refinement subtyping} $\JGSSubtype{r}{r'}$: subtyping between two non-recursive refinements.
\end{enumerate}

Nominal type subtyping follows the nominal subtyping graph (Definition \ref{defn:nominal_subtyping_graph}). To check $n_1 \, \overline{r_1}$ \subtypes $n_2 \, \overline{r_2}$, we first check if the two are related by the nominal subtype relation by finding a path from $n_1$ to $n_2$ in the nominal subtyping graph. Due to conditional subtyping, we must check if there is a path such that the refinement labelled on every edge along the path (the ``conditions'') are each satisfied by $\overline{r_1}$. Finally, we have to check if $\overline{r_2}$ still supertypes the LHS since it is possible that $\overline{r_2}$ makes some type member of $n_2$ too specific.




Member type subtyping is similar to DOT. Follow the upper bound for LHS base types, and follow the lower bound for RHS base types. Reflexivity applies for when base types on both sides are exactly the same, in which case structural subtyping applies to the refinements. If any one side's base type becomes a name type, it waits for the other side to also reduce into a name type, at which point nominal type subtyping applies.

Structural subtyping follows standard width and depth subtyping on record types.

\subsection{Member Subtyping Judgments}
\label{sec:subtyping:top-level}

\begin{figure}
  \input{04-subtyping-top}
  \caption{Member Declaration Subtyping}
  \label{fig:nominal_wyvern_subtyping_cont}
\end{figure}

Figure \ref{fig:nominal_wyvern_subtyping_cont} defines the rules for subtyping member declarations $\sigma$ that are inside top-level name declarations.

Judgment form $\JGSSubtype{\overline{\sigma_1}}{\overline{\sigma_2}}$ checks that the type members, fields and methods in $\overline{\sigma_1}$ are structurally compatible with those in $\overline{\sigma_2}$.

Rule {\sc S-Top-Nil} applies when the right-hand side is empty: anything is structurally compatible with an empty object.

Rule {\sc S-Top-Type} applies when the right-hand side has a type member declaration.
Here, we find a corresponding type member in the left-hand side and use the judgment form for type members $\JGSSubtype{\delta}{\delta'}$, defined in Section~\ref{sec:decidability:judgments}, to check that both of the members' types and bounds ($\le$ or $\ge$ or $=$) are compatible.

Rule {\sc S-Top-Field} applies when the right-hand side has a field declaration.
We find the corresponding field in the left-hand side and use the judgment form for types $\JGSSubtype{\tau}{\tau'}$, which is also defined in Section~\ref{sec:decidability:judgments}.

Finally, rule {\sc S-Top-Method} applies when the right-hand side has a method declaration.
We find the corresponding method in the left-hand side and use the judgment form for types to check that the method parameter types are contravariant subtypes, and that the results are (covariant) subtypes.
The result types may refer to the method parameter $x$ as path-dependent types, so the method parameter $x$ is added to the environment when checking the method result; we assume that both methods can be renamed to have the same parameter name.
Note that, unlike previous work on Decidable Wyvern \cite{mackay:decidable_wyvern} and Kernel $F_{<:}$ \cite{pierce:bounded_quant_undecidable}, we allow full contravariance here; this full contravariance is only possible inside the top-level method declaration judgments, as method types cannot occur inside types themselves in Nominal Wyvern.

\subsection{Decidability}
\label{sec:decidability:decidability}

\begin{theorem}[Subtyping is decidable]
\label{thm:subtyping_decidable}
For any Nominal Wyvern program with a valid material/shape separation of types, $\JGSSubtype{\tau_1}{\tau_2}$ can always be determined with a finite-length derivation, and thus subtyping is decidable.
\end{theorem}

To prove Theorem \ref{thm:subtyping_decidable}, we define the notion of a ``lineage'' in the context of a subtype derivation. The idea is that given a subtype query ($\JGSSubtype{\tau}{\tau}$), a lineage $\mathcal{L}$ captures the trace a type goes through during the subtype derivation that starts with that initial query. Concretely, a lineage is a tree with types as nodes. The shape of the tree corresponds exactly to the derivation tree of the initial subtype query. Each subtyping derivation creates two lineages: an initial left-lineage rooted at the initial LHS type, and an initial right-lineage rooted at the initial RHS type. For each nested subtyping judgment in the derivation tree, the inner judgment's (deeper in the tree) LHS and RHS types are linked to the outer judgment's types depending on which rules were used between the inner and outer subtype judgments. In most cases, the inner type is added as a child of the outer type of the same side. However, if the judgments involve a {\sc S-T-Eq} or {\sc S-T-Ge} (i.e.\ the bound on the type in the structure involved a lower bound), the inner LHS type is added as a child of the outer RHS type (``the left lineage swings to the right''), and the inner RHS links to the outer LHS. These two rules are the only cases when the two lineages swap sides.

\begin{definition}[Lineage]
Given a subtype derivation tree rooted at $\JGSSubtype{\tau_{initl}}{\tau_{initr}}$, the two lineages of the derivation tree are each a tree.
Each subtype judgment in the derivation tree is given a label, and for each subtype judgment $J$: $\JGSSubtype{\tau_l}{\tau_r}$, two vertices are created: $J\#\tau_l$ and $J\#\tau_r$.
For each pair of judgments $J_1$: $\JGSSubtype{\tau_{l1}}{\tau_{r1}}$ and $J_2$: $\JGSSubtype{\tau_{l2}}{\tau_{r2}}$ such that $J_1$ is the closest ancestor of $J_2$ that is a subtype judgment {\sc S-*}, denote new sets of edges:
\begin{itemize}
  \item Covariant edges: $J_1\#\tau_{l1} \rightarrow J_2\#\tau_{l2}, J_1\#\tau_{r1} \rightarrow J_2\#\tau_{r2}$
  \item Contravariant edges: $J_1\#\tau_{l1} \rightarrow J_2\#\tau_{r2}, J_1\#\tau_{r1} \rightarrow J_2\#\tau_{l2}$
\end{itemize}
, with each edge labelled with the rules used for the recursive call. Then generate edges depending on the path from $J_1$ to $J_2$:
\begin{itemize}
  \item If $J_1$ calls to $J_2$ via a {\sc S-T-Eq}, add both contravariant and covariant edges
  \item If $J_1$ calls to $J_2$ via a {\sc S-T-Ge}, add contravariant edges
  \item Otherwise, add covariant edges
\end{itemize}
The initial left-lineage and right-lineage are the trees rooted at $\tau_{initl}$ and $\tau_{initr}$, respectively.
\end{definition}

\begin{example}
Consider the following derivation (each judgment given a label $J_n$). The two corresponding lineages are the two trees (each with only one branch throughout) on the right. If $n_2 \subtypes n_3$ is supported by the nominal subtyping graph, the initial judgment is proven true.
\begin{footnotesize}
\[
  \begin{array}{rrlll}
  \Gamma \vdash &x_1.t_1 &\subtypes x_2.t_2 \{t_3 \ge n_2 \} &\quad [J_1] \\
  \Gamma \vdash &x_2.t_2 \{t_3 \ge n_3 \} &\subtypes x_2.t_2 \{t_3 \ge n_2 \} &\quad [J_2] \\
  \Gamma \vdash &n_2 &\subtypes n_3 &\quad [J_3]
  \end{array}
\]

where
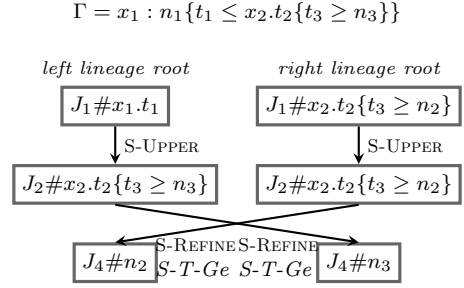
\[
\Gamma = x_1 : n_1 \{t_1 \le x_2.t_2 \{t_3 \ge n_3\} \}
\]
\end{footnotesize}


\begin{footnotesize}
\[
\begin{tikzpicture}[node distance=5mm]
  \node[squarednode]      (l2)                       {$J_2\#x_2.t_2 \{t_3 \ge n_3 \}$};
  \node[squarednode]      (r2)       [right=of l2]   {$J_2\#x_2.t_2 \{t_3 \ge n_2 \}$};
  \node[squarednode]      (l3)       [below=of l2]   {$J_4\#n_2$};
  \node[squarednode]      (r3)       [below=of r2]   {$J_4\#n_3$};

  \node[squarednode]      (l1)       [above=of l2]   {$J_1\#x_1.t_1$};
  \node[squarednode]      (r1)       [above=of r2]   {$J_1\#x_2.t_2 \{t_3 \ge n_2 \}$};

  \node[draw=none,fill=none]         [above=of l1, yshift=-5mm]   {\scriptsize left lineage root};
  \node[draw=none,fill=none]         [above=of r1, yshift=-5mm]   {\scriptsize right lineage root};
    
  \draw[arrow] (l1.south) -- (l2.north) node[midway, right] {\scriptsize {\sc S-Upper}};
  \draw[arrow] (r1.south) -- (r2.north) node[midway, right] {\scriptsize {\sc S-Upper}};
  \draw[arrow] (l2.south) -- (r3.north) node[pos=.67, below,align=left] {\scriptsize \sc S-Refine\\S-T-Ge};
  \draw[arrow] (r2.south) -- (l3.north) node[pos=.67, below,align=left] {\scriptsize \sc S-Refine\\S-T-Ge};
\end{tikzpicture}
\]
\end{footnotesize}


\end{example}

A lineage captures the relation between types in recursively dependent subtyping judgments. As long as all paths in a lineage tree are finite, the entire corresponding subtype derivation is finite. We consider any path starting from the root of a lineage as made up of many segments that are divided by edges caused by {\sc S-NameUp} derivations. Below, we first study the behaviour within a segment and then extend across segments.

To prove decidability, we define a measure $\mathcal{E}$ on types that will decrease during derivation. $\mathcal{E}$ can be thought of as the ``potential energy'' of a type, and that continued subtype derivation requires spending energy. To define $\mathcal{E}$, we first define two measures, $\mathcal{M}$ and $\mathcal{A}$, on type members of a given name type $n$. $\mathcal{M}(n \dcolon t)$ captures the dependence of $t$ on other type members of $n$. $\mathcal{A}(n \dcolon t)$ captures the the dependence of $t$ on named types.

\begin{definition}[$\mathcal{M}$ and $\mathcal{A}$ measures of pseudotypes]
Given a named type definition $n : \{ x \Rightarrow \overline{\sigma} \} \in \Delta$ and the subtype dependency graph derived from $\Delta \Sigma$, the $\mathcal{M}_{\Delta\Sigma}$ and $\mathcal{A}_{\Delta\Sigma}$ measures of a pseudotype $n \dcolon t$ is defined as:
\begin{footnotesize}
\begin{mathpar}
  \inferrule
  {
  \Type{t}{B}{\tau} \in \overline{\sigma}\\
  T = \{ n \dcolon t' | x.t' \in \tau \land n \dcolon t \rightarrow_B n \dcolon t' \}
  }
  {\tmultiplier{n \dcolon t} = 1 + \sum_{x.t' \in T}{\tmultiplier{n \dcolon t'}}}

  \inferrule
  {
  \Type{t}{B}{\tau} \in \overline{\sigma}\\
  T = \{ n \dcolon t' | x.t' \in \tau \land n \dcolon t \rightarrow_B n \dcolon t' \}
  }
  {\tadder{n \dcolon t} = 1 + \sum_{x.t' \in T}{\tadder{n \dcolon t'}} + \sum_{n' \in \tau}{\energy{n'}}}
\end{mathpar}
\end{footnotesize}
where $n \dcolon t \rightarrow_B n \dcolon t'$ is true if there is a path in the subtype dependency graph from $n \dcolon t$ to $n \dcolon t'$ that only consists of edges whose variance is $B$ and are not labeled with shapes, and the path does not consist of nodes that are not pseudotypes of $n$.

\noindent
$\energy{}$ is the measure on types to be defined below.

\end{definition}

\begin{definition}[Energy measure of a type]
Given the contexts $\Delta \Sigma \Gamma S$, first define the energy $\mathcal{E}_{\Delta \Sigma \Gamma S}$ of a \textbf{base type} as:
\begin{footnotesize}
\begin{mathpar}
  \inferrule
  {\\}
  {\energy{\top} = 0}
\and
  \inferrule
  {\\}
  {\energy{\bot} = 0}
\and
  \inferrule
  {
  N = \{ n_1 | n_1 r_1 \subtypes n \in \Sigma \} \cup \{ n' | n_1 r_1 \subtypes n \in \Sigma, n' \in r_1 \}
  }
  {\energy{n} = \sum_{n' \in N}(\energy{n'}) + 1}
\and
  \inferrule
  {
  \JGSType{p}{\tau_p}\\
  \JGSExpose{\tau_p}{n \, r}\\
  }
  {\energy{p.t} = \energy{n \, r} \times \tmultiplier{n \dcolon t} + \tadder{n \dcolon t}}
\end{mathpar}
\end{footnotesize}
The energy $\mathcal{E}_{\Delta \Sigma \Gamma S}$ of a type $\tau$ is the sum of the energies of the base types that appear in $\tau$.
\begin{footnotesize}
\[
  \energy{\tau} = \sum_{\beta \in \tau}{\energy{\beta}}
\]
\end{footnotesize}
\end{definition}

The material/shape separation requirements ensure that the energy measurement is defined on name types: The rules dictate a topological order over all named types, which makes sure that it is always possible for the recursive energy calculation to terminate for a valid subtyping dependency graph.

The intuition behind the energy equation of path-dependent types is to make sure that even if one type member $t$ depends on many other type members, the energy of $p.t$ for any $p$ should be greater than the combined energy of all that $t$ depends on in its named type. Due to the existence of cyclic dependencies through shapes, this is not always possible to compute. However, the definition of $\mathcal{M}$ and $\mathcal{A}$ takes this into account by ignoring any dependencies that are inside the refinement of a shape. This means during subtype derivation, when following the type bound of $t$, as long as the bound does not include any shapes, the energy after will be smaller than before. 

By casing on the judgment rules, we can see that within any segment of a lineage, the energy of nodes do not increase going downwards when the lineage is on the RHS. The same is true for the LHS, but since we allow certain cyclical dependencies to exist through shapes, the LHS energy can suddenly increase when a shape becomes a base type after applying {\sc S-Upper} to unfold the LHS. The challenge thus becomes limiting the number of shapes. Material/shape separation rule 1 outlaws shapes from appearing inside lower bounds or equalities, which ensures that unfolding path-dependent types on the RHS cannot result in a shape on the RHS. This restriction limits shapes on the RHS to those that already exist on the RHS from the beginning of the segment (i.e. no shapes can be brought over from the LHS or generated on the RHS). Thus, energy increase cannot go on forever because to increase a second time requires eliminating the earlier LHS shape with {\sc S-Refine}, which requires consuming a shape from the RHS. Therefore, any infinite derivation cannot be entirely within one segment: Both lineages must eventually go through {\sc S-NameUp}.

Rule {\sc S-NameUp} contains two recursive subtyping derivations: one for the sub-refinement $r \subtypes r_\Sigma$, and one for the named supertype $n_\Sigma~r \subtypes n'~r'$. The former decreases the energy on both sides by recursing into strictly smaller types.
The latter reduces the energy on the LHS, as the energy of the original LHS named type $n$ includes the supertype $n_\Sigma$ in its sum.
The key, however, is that neither replenishes any recursive uses of shapes for the lineages. Material/shape separation rule 3, which outlaws shapes from being refined \emph{inside refinements}, guarantees that any new shape mentioned in $r_\Sigma$ will not be refined, thus ending derivation immediately when encountered. Without RHS shape replenishing, eventually, both lineages will have no shapes to use for energy increase on the RHS.

\begin{lemma}[Lineages have finite depth]
\label{thm:lineage_finite}
For any Nominal Wyern program with a valid material/shape separation of types, the two lineages of any subtype derivation both have finite depth.
\end{lemma}

Due to the combined result of energy decreasing and the number of RHS shapes decreasing, no lineage can have infinite depth (Lemma \ref{thm:lineage_finite}). Thus, all subtype derivations must eventually terminate, proving Theorem \ref{thm:subtyping_decidable}.

\subsection{Increasing Expressivity via Expansion}
\label{sec:increasing-expressivity:extension}

The system described so far is decidable, but is somewhat limited in terms of expressivity.
Furthermore, there are some subtyping relationships that intuitively should hold, but do not.
To address this issue, we perform a pre-processing expansion phase which partially unfolds types before using the decidable subtyping rules.
This pre-processing is itself decidable, which ensures that the whole system remains decidable.

\begin{lstlisting}[mathescape, style=nwyvern, label={lst:pdt_primer2}, caption={Expressivity Example}]
name List { this =>
	type T <= Top
}

name IntList { this =>
	type T  = Int
}
subtype IntList <: List

assert IntList <: List { type T = Int }
\end{lstlisting}

As a concrete example, consider the program in \autoref{lst:pdt_primer2}, which defines the \code{List} type with a type member for the element type and a specialised \code{IntList} type with the elements specialised to \code{Int}.
The program declares that \code{IntList} is a subtype of \code{List}, which succeeds with the system shown previously.
However, when we then assert that the type \code{IntList} is itself a subtype of the refined type \code{List \{ type T = Int \}}, the subtyping relation does not hold.

This behaviour does not match our intuition of refinements or type members: we already know that \code{IntList} is a \code{List}, and we know that \code{IntList}'s \code{T} member is an \code{Int}, so it seems reasonable to expect the assertion to hold.
However, the issue is that the subtyping system ignores the left-hand-side's base type when checking refinements, as in the following (stuck) derivation tree:

\begin{mathpar}
    \inferrule
        {\inferrule
            {\{ \} \subtypes \{ \Type{T}{=}{\code{Int}} \}}
            {\code{List}~\{ \} \subtypes \code{List}~\{ \Type{T}{=}{\code{Int}} \}}
            {\JLabel{S-Refine}}
        }
        {\code{IntList}~\{ \} \subtypes \code{List}~\{ \Type{T}{=}{\code{Int}} \}}
        {\JLabel{S-NameUp}}
\end{mathpar}

This tree is stuck: the {\sc S-R-Cons} rule cannot apply at the top, as the left-hand-side refinement has no matching definition for the \code{F} type member.

To resolve this issue, we modify the typechecking process to \emph{expand} the types, unfolding the type members to a certain depth, before any subtyping checks.

In the above example, the expanded subtyping check would be \code{IntList \{ type T = Int \} <: List \{ type T = Int \}}, which does succeed.

\begin{figure}
\input{extension-expansion}
\caption{Expansion definitions}
\label{fig:expansion_def}
\end{figure}

\autoref{fig:expansion_def} shows the definition of expansion.
The functions here are parameterised by the environments $\Gamma$ and $\Delta$; we omit these parameters to reduce clutter.
Function \emph{depth} computes the maximum refinement depth of a type.
Function \emph{expand} recursively expands a type to the given depth.
Function $\textit{extend}_1$ performs a single step of expansion; named types are looked up and their type members are unfolded.
The type members inside a named type definition may refer to the self type $z$, so this binding must in turn be unfolded (\autoref{sec:design:avoid}).
Finally, function \emph{check} performs the original subtyping check after expanding both types to the maximum depth.

In our extended version of the typing rules, all subtyping checks, except those performed by subtyping itself, are replaced with calls to this check function.

\section{Type Safety}
\label{ch:safety}

We now prove that Nominal Wyvern is type safe.
The technique used to prove type safety uses a fuel-annotated big-step semantics, as in \cite{rompf:dot_soundness,amin:foundations_pdt}.
However, the Nominal Wyvern proofs differ from these works in two main ways.
First, our grammar of types is simpler, and does not have any binding forms.
Removing binding forms from the type grammar greatly simplifies the proof of subtyping transitivity and inversion of subtyping.
(Top-level type definitions do have binding forms for the self-variable and method arguments, but unlike DOT, these do not occur in types $\tau$.)
Secondly, our term typing does not have subsumption of terms, only on values.
Removing subsumption simplifies the inversion of typing, but restricts our proofs to using a big-step semantics rather than a small-step semantics.

Our type safety argument considers the base Nominal Wyvern without the expansion extension presented previously in \autoref{sec:increasing-expressivity:extension}.

\subsection{Dynamic Semantics}
\label{sec:grammar:reduction}

Figure \ref{fig:reduction} presents the big-step dynamic semantics of Nominal Wyvern. The ``runtime'' includes a heap storage $\mu$ that stores memory locations. Each memory location $l$ contains the definition of an object created via \code{new}.
The result of evaluation is a result heap $\mu'$ and a location $l'$; the returned location may refer to values in the result heap.

\begin{figure}[h]
  \input{05-reduction}
  \caption{Nominal Wyvern Reduction Rules}
  \label{fig:reduction}
\end{figure}

The evaluation rules follow the straightforward interpretation of the expressions. Memory locations represent real objects during runtime, and are treated as values: evaluation stops when the entire program reduces into a location.

Rule {\sc Ev-Loc} denotes that a location $l$ reduces to itself without modifying the heap.

Rule {\sc Ev-Field} denotes evaluation of a field access on an object $l_s$. The field definition $p_v$ inside the object may refer to the object itself using variable $x_s$; in the returned value, we substitute the self variable $x_s$ with the actual object location $l_s$.

Rule {\sc Ev-Method} denotes evaluation of a method $l_s.f$ applied to the argument $l_a$. The target object $l_s$ has a method definition containing the expression for the method body $e_f$, which may refer to the self variable $x_s$ and the method argument $x_a$. We evaluate this method, substituting the target object and argument.
Note that the big-step reduction evaluates the method body to a value $l'$, but requires the target object and method argument to already be locations, as the program must be in a-normal form.

Rule {\sc Ev-New} allocates a new object by generating a fresh location and adding the object definition to the heap.

Rule {\sc Ev-Let} first evaluates the definition $e_1$ to a location $l_1$, and then evaluates the continuation $e_2$ with any references to the definition substituted in.

The semantics presented here differ slightly from \cite{amin:foundations_pdt}, which defines values as a closure consisting of a heap and a location; unpacking the closure and enforcing a single global heap simplifies the treatment of path-dependent methods, whose result type may refer to the heap.
Note also that in all substitutions, the right-hand-side is limited to a location, which is a form of path.
This limited form of substitution ensures that substitutions in path-dependent types preserve grammatical validity.

\subsection{Type Safety Proof}
\label{sec:proof}

As in \cite{amin:foundations_pdt}, we prove preservation on the big-step semantics, and progress on a fuel-annotated big-step semantics.
We include the proof of preservation and major lemmas here. The proof of progress and minor lemmas are deferred to \autoref{ch:appendix}.
The main lemmas we require for preservation are transitivity of subtyping, subsumption of values (if $\JType{\Gamma}{S}{l}{\tau}$ and $\tau <: \tau'$ then $\JType{\Gamma}{S}{l}{\tau'}$), and that substitution preserves typing and subtyping.

\begin{theorem}[Preservation of types]
  For a well-typed heap $\JJudge{\cdot}{S}{\mu}$ and well-typed expression $\JType{\cdot}{S}{e}{\tau}$,
  if $\StepsTo{\mu}{e}{\mu'}{l'}$
  then there exists some context $S'$ such that
  $\JJudge{\cdot}{S'}{\mu'}$ and $\JType{\cdot}{S'}{l'}{\tau}$.
\end{theorem}
\begin{proof}
  By induction on the reduction relation $\StepsTo{\mu}{e}{\mu'}{l'}$:
  \begin{description}
    \item[Case] {\sc Ev-Loc}: \\
      Exists $S$; by $e$ well-typed assumption.
    \item[Case] {\sc Ev-Field}: \\
      Exists $S$; by $\mu$ well-typed assumption.
    \item[Case] {\sc Ev-Method}: \\
      By reduction, have: $e = l_s.f(l_a)$; and $\mu(l_s) \ni_{x_s} \DFun{f}{\tau_a'}{x_a}{\tau_r'}{e_f}$; and $\StepsTo{\mu}{e_f[x_s := l_s, x_a := l_a]}{\mu'}{l'}$. \\
      By $e$ well-typed, have: $\JType{\cdot}{S}{l_s}{\tau_s}$; and $\tau_s$ exposes to $n~r$; and method $\Fun{f}{\tau_a}{x_a}{\tau_r}$ in $n~r$; and argument $l_a: \tau_{l_a}$; and $\tau_{l_a} \subtypes \tau_a$; and $\tau = \tau_r[x_s := l_s, x_a := l_a]$. \\
      By $\mu$ well-typed, runtime type of method is a subtype of static type: $\Fun{f}{\tau_a'}{x_a}{\tau_r'} \subtypes \Fun{f}{\tau_a}{x_a}{\tau_r}$. \\
      By $\mu$ well-typed, method body is well-typed: $\JType{x_s: \tau_s', x_a: \tau_a'}{S}{e_f}{\tau_r'}$. \\
      By lemmas subsumption of values and transitivity of subtyping, argument matches expected type: $\JType{\cdot}{S}{l_a}{\tau_{l_a} \subtypes \tau_a \subtypes \tau_a'}$.\\
      By lemma substitution preserves typing, have: $\JType{\cdot}{S}{e_f[x_s := l_s, x_a := l_a]}{\tau_r'[x_s := l_s, x_a := l_a]}$. \\
      By inductive hypothesis, obtain $S'$ such that $\JType{\cdot}{S'}{l'}{\tau_r'[x_s := l_s, x_a := l_a]}$. \\
      By lemma substitution preserves subtyping, have: $\JSubtype{\cdot}{S'}{\tau_r'[x_s := l_s, x_a := l_a]}{\tau_r[x_s := l_s, x_a := l_a]}$. \\
      Hence, exists $S'$ such that $\JType{\cdot}{S'}{l'}{\tau}$.
    \item[Case] {\sc Ev-New}: \\
      Exists $S, l: \{ x_s \Rightarrow \text{sig}(\overline{d}) \}$; by $e$ well-typed assumption.
    \item[Case] {\sc Ev-Let}: \\
      By reduction, have: $e = \Let{x}{e_1}{e_2}$; and $\StepsTo{\mu}{e_1}{\mu'}{l_1}$; and $\StepsTo{\mu'}{e_2[x := l_1]}{\mu''}{l_2}$. \\
      By $e$ well-typed, have $\JType{\cdot}{S}{e_1}{\tau_1}$; and $\JType{x: \tau_1}{S}{e_2}{\tau_2}$; and $\JAvoid{x: \tau_1}{\tau_2}{\le}{x}{\tau}$. \\
      By inductive hypothesis, obtain $S'$ such that $\JType{\cdot}{S'}{l_1}{\tau_1}$. \\
      By lemma big-step-heap-only-grows, have $\mu$ is a prefix of $\mu'$ and $S$ is a prefix of $S'$. \\
      By lemma weakening, extend typing of $e_2$ to $\JType{x: \tau_1}{S'}{e_2}{\tau_2}$. \\
      By substitution of $l_1$ for $x$ in $e_2$, have $\JType{\cdot}{S'}{e_2[x := l_1]}{\tau_2[x := l_1]}$. \\
      By lemma avoid-is-supertype, $\JSubtype{x: \tau_1}{S'}{\tau_2}{\tau}$. \\
      By lemma substitution preserves subtyping, have $\JSubtype{\cdot}{S'}{\tau_2[x := l_1]}{\tau[x := l_1]}$. \\
      By lemma avoid-binding, $x$ does not occur in $\tau$; thus $\JSubtype{\cdot}{S'}{\tau_2[x := l_1]}{\tau}$. \\
      By lemma subsumption of values, $\JType{\cdot}{S'}{e_2[x := l_1]}{\tau}$. \\
      By inductive hypothesis, obtain $S''$ such that $\JType{\cdot}{S''}{l_2}{\tau}$. \\
      Hence, exists $S''$ such that $\JType{\cdot}{S''}{l_2}{\tau}$.
  \end{description}
\end{proof}

\subsubsection{Auxiliary lemmas}

\begin{lemma}[Transitivity of subtyping]
  For a well-formed context $\Gamma$ and correspondingly well-formed types $\tau_1$, $\tau_2$ and $\tau_3$,
  if $\tau_1 \subtypes \tau_2$ and $\tau_2 \subtypes \tau_3$, then $\tau_1 \subtypes \tau_3$.
\end{lemma}
The proof of transitivity is straightforward compared to DOT, as types do not contain binding forms.
The proof uses induction over the subtyping derivation tree; the proof itself is presented in \autoref{ch:appendix}.

\begin{lemma}[Subsumption of value typing]
  For a given heap $\mu$ with binding $l$, if context $S$ is well-typed $\JJudge{\cdot}{S, l: \tau}{\mu}$, the binding can be subsumed to a supertype $\tau \subtypes \tau'$ such that $\JJudge{\cdot}{S, l: \tau'}{\mu}$.
\end{lemma}
\begin{proof}
  By definition of heap typing judgment (\autoref{fig:reduction}) and value typing (\autoref{fig:nominal_wyvern_term_typing}), there is some type $\tau_l$ such that $\mu(l): \tau_l$ and $\tau_l \subtypes \tau$. \\
  By transitivity of subtyping, $\tau_l \subtypes \tau \subtypes \tau'$.
\end{proof}

\begin{lemma}[Big-step heap only grows]
  If an expression $e$ reduces $\StepsTo{\mu}{e}{\mu'}{l'}$, then $\mu$ is a prefix of $\mu'$.
  Similarly, for a context $S$ and $S'$ such that $\JJudge{\cdot}{S}{\mu}$ and $\JJudge{\cdot}{S'}{\mu'}$, then $S$ is a prefix of $S'$.
\end{lemma}

\begin{proof}
  By induction on reduction relation.
  Most rules return the heap as-is, while rule {\sc Ev-New} allocates a fresh location which is appended to the end of the heap.
  The context grows correspondingly.
\end{proof}

\begin{lemma}[Location-for-variable substitution preserves exposure]
  For an environment pair $\Gamma$ and $S$, for a variable $x: \tau_x \in \Gamma$, heap location $l: \tau_x \in S$, initial type $\tau$ and exposed type $\tau'$ such that $\JGSExpose{\tau}{\tau'}$, then the substituted exposure relation
  $\JGSExpose{\tau[x := l]}{\tau'[x := l]}$ also holds.
\end{lemma}
\begin{proof}
  By induction on the exposure judgment.
  Rules {\sc Exp-Top, Exp-Bot, Exp-Name} and {\sc Exp-Otherwise} are trivial; rule {\sc Exp-Upper} follows from application of inductive hypothesis.
\end{proof}

\begin{lemma}[Location-for-variable substitution preserves subtyping]
  For an environment pair $\Gamma$ and $S$, for a variable $x: \tau_x \in \Gamma$, heap location $l: \tau_x \in S$, subtype $\tau$ and supertype $\tau'$ such that $\JSubtype{\Gamma}{S}{\tau}{\tau'}$, then the substituted subtyping relation
  $\JType{\Gamma}{S}{\tau[x := l]}{\tau'[x := l]}$ also holds.
\end{lemma}
\begin{proof}
  By mutual induction on the subtyping relations for types, refinements and type members.
  Rules {\sc S-Top, S-Bot} and {\sc S-R-Nil} are trivial.
  Rules {\sc S-NameUp, S-T-Eq, S-T-Le, S-T-Ge} and {\sc S-R-Cons} follow from application of inductive hypothesis.
  Rules {\sc S-Lower} and {\sc S-Upper} use upcasting and downcasting, respectively, and require the use of substitution-preserves-exposure lemma in addition to inductive hypothesis.
\end{proof}

The statement of the corresponding lemma for substitution on typing judgment requires some care, as the type may contain path-dependent variables.
For a typing judgment $\JGSType{e}{\tau}$, when we substitute into the expression $e[x := l]$, we need to perform the substitution on the type as well: $\tau[x := l]$.
However, this type-substitution is still not quite enough: there may be variables in the environment whose types refer to $x$ as well, so we substitute into the environment too.

\begin{lemma}[Location-for-variable substitution preserves typing]
  For an environment pair $\Gamma$ and $S$, if $x: \tau_x \in \Gamma$ and $l: \tau_x \in S$, and given a well-typed expression $e$ such that $\JType{\Gamma}{S}{e}{\tau}$, then
  $\JType{\Gamma[x := l]}{S[x := l]}{e[x := l]}{\tau[x := l]}$.
\end{lemma}
\begin{proof}
  By induction on the typing derivation.
  Cases follow from application of substitution-preserves-exposure and substitution-preserves-subtyping lemmas, as well as inductive hypothesis.
\end{proof}

\begin{lemma}[Avoid ``avoids'' binding]
  For an environment pair $\Gamma$ and $S$ and types $\tau$ and $\tau'$, if $\Avoid{\tau}{B}{x}{\tau'}$, then $x$ is not free in $\tau'$.
\end{lemma}
\begin{proof}
  By straightforward induction on the avoidance derivation.
\end{proof}

\begin{lemma}[Avoid is supertype]
  For an environment pair $\Gamma$ and $S$, if $x: \tau_x \in \Gamma$ and $\Avoid{\tau}{B}{x}{\tau'}$, then $\tau$ and $\tau'$ are related by $B$'s corresponding subtype relation; that is, $\tau$ is a subtype of $\tau'$ when $B$ is ($\le$), a supertype when $B$ is ($\ge$), or both when $B$ is ($=$).
\end{lemma}
\begin{proof}
  By induction on the avoidance derivation.
  Cases {\sc Avoid-Top, Avoid-Bot, Avoid-Name, Avoid-Path-NE} and {\sc Avoid-Eq} follow from reflexivity of subtyping.
  Cases {\sc Avoid-Path-Eq} requires case analysis on $B$; for $\le$, the subtype relation is proved with {\sc S-Lower}; for $\ge$, the subtype relation is proved with {\sc S-Upper}.
\end{proof}


\subsubsection{Progress of type safety}

The above proof of preservation is only one half of the type safety proof, but contains the most salient details.
The other half, the \emph{progress lemma}, uses a fuel-annotated variant of the reduction rules.
The fuel-annotated rules and the proof are presented in \autoref{ch:appendix}.

\section{Expressiveness}
\label{ch:expressiveness}
\subsection{F-Bounded Polymorphism}
\label{sec:expr:f-bounded}

Recall from section \ref{sec:bg:undecidable_subtyping} that one of the benefits of combining subtype polymorphism with parametric polymorphism is the ability to express F-bounded polymorphism. The \code{Cloneable} example of \ref{fig:f_bounded_poly_ex} demonstrates a positive recursion, in which the recursive type variable is at an ``output'' position, or covariant. When positive recursive usages are encoded in mainstream object-oriented languages that do not support F-bounded polymorphism, the output type is usually the most general type, and a dynamic cast is performed to get back the original type. With F-bounded polymorphism, the type of the output can be guaranteed statically, as shown in Listing \ref{lst:expr:f-bounded_cloneable}.

\begin{lstlisting}[mathescape, style=nwyvern, label={lst:expr:f-bounded_cloneable}, caption={F-bounded polymorphism in Nominal Wyvern}]
@shape name Cloneable {z =>
  type t <= $\top$
  def clone : unit _ -> z.t
}
name String {z =>
  type t <= String
  def clone : unit _ -> z.t
}
subtype String <: Cloneable
...
...(in some object definition)... {
  def makeClone : Cloneable arg -> arg.t = arg.clone ()
}
\end{lstlisting}

Figure \ref{fig:nominal_wyvern_syntax_eg} shows that Nominal Wyvern can also support negative recursion: When the parameterized type is an input to a method, or contravariant. Without proper support, users have to resort to dynamically checking the input type before casting it (as is commonly seen in Java's \code{equals()} methods), whereas in Nominal Wyvern, the type system is able to guarantee a compatible input statically. Furthermore, mutually dependent set of types (aka. family polymorphism \cite{ernst:family_poly}) can also be statically checked in Nominal Wyvern.

\subsection{Representing ML Modules}
\label{sec:expr:ml_module}

Data abstraction in ML is based on abstract data types (ADT). An ADT encapsulates an abstract type along with operations on the type.
Formally, ADTs are modeled with existential types: $\exists t. \tau$, where $\tau$ is typically a product of functions that operate on the abstract type $t$. This can simply be represented in DOT-based systems by an object with a type member.

Note that the functions of an ADT are coupled with the type. ML modules solve this single implementation problem by wrapping them in named structures. Signatures define interfaces, and structures ascribe to signatures and define their own implementation. This can be very closely modeled by Nominal Wyvern. A module ascribing to a signature in ML corresponds to an object exhibiting a type in Nominal Wyvern. Functors naturally come for free as well.
\subsection{Object-Oriented Programming}
\label{sec:expr:oop}

Compared to ADTs, objects do not provide any types. They instead provide implementations for a common type with a common interface. In fact, multiple objects of the same type can have wildly different implementations. Yet they can still interact with each other with no regard to the internal differences since, instead of unpacking the implementation type, they only rely on dynamically dispatched method calls over the common interface. This added interoperability contributes to the success of OOP languages \cite{aldrich:interoperability}.

Nominal Wyvern's semantic separation naturally supports a pure OOP approach: Named types serve as interface definitions, and objects created from named types serve as constructors, or ``classes''. This way, the syntax guarantees interfaces are not tied to any implementation, and classes are syntactically different constructs than types. Classes are thus able to serve as pure organizers of implementations. Listing \ref{lst:expr:oop_pure} translates the sets example from \citet{cook:understanding_data_abstraction_revisited} into Nominal Wyvern. In Cook's paper, \code{ISet} defines the interface for sets, while the classes are simply constructor functions. Once created, an object is no longer associated with its constructing class, and can be freely used with objects created from other classes.

To simplify the presentation, we allow using tuple types when defining or calling methods below. They are considered named types that are automatically declared at the top level.

\begin{lstlisting}[mathescape, style=nwyvern, label={lst:expr:oop_pure}, caption={Pure OOP in Nominal Wyvern}]
// Assume pre-defined Int and Bool types.
// Int has builtin constructors from literals, and an equals() method.
// Bool has builtin constructors "true" and "false", binary operator "||"
// for logical or, and method if(t,f) that returns t if true or else f.

// interface for sets
name ISet {s =>
  def isEmpty() : Bool
  def contains(i: Int) : Bool
  def insert(i: Int) : ISet
  def union(s: ISet) : ISet
}
// define classes/constructors
name SET_CONS {c =>
  def Empty() : ISet
  def Insert(s: ISet, n: ISet) : ISet
  def Union(s1: ISet, s2: ISet) : ISet
}

let Set = new SET_CONS {c =>
  def Empty() =
    new ISet {z =>
      def isEmpty() = true
      def contains(i: Int) = false
      def insert(i: Int) = c.Insert(z, i)
      def union(s: ISet) = s
    }
  def Insert(s: ISet, n: Int) =
    s.contains(n).if(s, new ISet {z =>
      def isEmpty() = false
      def contains(i: Int) = (i.equals(n)) || (s.contains(i))
      def insert(i: Int) = c.Insert(z, i)
      def union(s: ISet) = c.Union(z, s)
    })
  def Union(s1: ISet, s2: ISet) =
    new ISet {z =>
      def isEmpty() = s1.isEmpty() || s2.isEmpty()
      def contains(i: Int) = (s1.contains(i)) || (s2.contains(i))
      def insert(i: Int) = c.Insert(z, i)
      def union(s: ISet) = c.Union(z, s)
    }
} in

let s1 = Set.Empty() in              // {}
let s2 = Set.Insert(s1, 1) in        // {1}
let s3 = s1.insert(2) in             // {2}
let s4 = Set.union(s2,s3) in         // {1,2}
...
\end{lstlisting}

\section{Related Work and Discussion}
\label{ch:related}

Nominal Wyvern arrives amid significant progress in the DOT community. Despite the problem of ``bad bounds'' (when a type member's upper bound is not a supertype of its lower bound), DOT was finally proven type safe \cite{rompf:dot_soundness,rapoport:dot_soundness}. Even though Nominal Wyvern syntactically restricts type bounds to be one-sided, the intersection semantic of refinements means it is still possible for abstract types to be bounded by bad bounds. However, similar to DOT, type safety is achieved because types with bad bounds only exist abstractly and can never be instantiated due to the stricter requirements that guard object creation from the more lenient world of abstract types.

More recently, \citet{rapoport:a_path_to_dot} introduced general paths to DOT with pDOT and proved it sound. Prior to pDOT, the provably type safe versions of DOT only allowed paths of length zero (i.e. a lone variable with no field accesses). This is a notable step towards a more Scala-like foundational language. Building on top of pDOT's soundness proof, Nominal Wyvern also supports paths of any length, and takes this into account when proving subtype decidability.

On the subtype decidability front, \citet{nieto:algorithmic_typing_dot} introduced decidable algorithmic typing and subtyping to a version of D\textsubscript{<:}, a subset of DOT without recursive types. Similarly in spirit, the typing and subtyping rules of Nominal Wyvern are completely syntax directed. In fact, the proof of subtype decidability already views the subtyping judgment rules as an algorithm with mutually recursive procedures. Nominal Wyvern differs from this version of D\textsubscript{<:} partly due to support for recursive types, which are critical to the practical expressiveness of a type system. More recently, \citet{hu:decidable_dsub_fragments} introduced Kernel D\textsubscript{<:} and Strong Kernel D\textsubscript{<:}, two decidable fragments of D\textsubscript{<:} derived by modifying/removing certain typechecking rules that are deemed problematic. Nominal Wyvern differs from both modifications of D\textsubscript{<:} mostly due to its nominal nature and the use of of material/shape separation, a non-typechecking-based approach that works with nominality to achieve decidability. Nominal Wyvern's support for recursive types and use of material/shape separation makes it much closer to Decidable Wyvern, which brings subtype decidability to a more complete DOT-like system. An alternative approach was proposed by \citet{mackay:popl2020}, achieving decidability by applying material/shape separation \cite{greenman:shapes} on top of DOT. While this work achieves decidability for key expressive features of DOT, the additional rules and syntax variation needed to enforce the material/shape separation add to the complexity of the type system. As explained in section \ref{sec:design:compare_decidable_wyvern}, Nominal Wyvern improves upon Decidable Wyvern by introducing nominality, which brings 1) a simpler material/shape separation rules, and 2) technical properties that relate to ease of use.

\section{Conclusion}
\label{ch:conclusion}

This paper presents Nominal Wyvern, a nominal, dependent type system that takes a novel approach to merging structural and nominal types. Nominal Wyvern achieves a higher degree of nominality by semantically separating the definition of structures and their subtype relations from arbitrary width refinements and the type bounds. This contributes to a system with more explicit meanings and relations, useful for both human readers to reason about and programming tools to refer to. Nominality also helps with achieving subtype decidability. In line with the theme of semantic separation, Nominal Wyvern adapts material/shape separation so that decidability results from an intuitive separation of types with different roles. This contributes to a restriction that is more easily understandable and articulable. The resulting system preserves the ability to express common patterns expressible with DOT, at the same time allowing for patterns that will be familiar to programmers already used to traditional functional or object-oriented programming languages.

\bibliographystyle{plainnat}
\bibliography{references,alex} 

\appendix
\section{Appendix}
\label{ch:appendix}
\subsection{Transitivity of subtyping}

The proof of transitivity is straightforward compared to DOT, as types do not contain binding forms.

\begin{theorem}[Transitivity of subtyping]
  For a well-formed context $\Gamma$ and correspondingly well-formed types $\tau_1$, $\tau_2$ and $\tau_3$,
  if $\tau_1 <: \tau_2$ and $\tau_2 <: \tau_3$, then $\tau_1 <: \tau_3$. \\
  Subtyping of refinements is similarly transitive.
\end{theorem}

\begin{proof}
  We will perform the proof by mutual induction over the subtyping derivation tree and the refinement subtyping derivation.
  We start by induction on the derivation tree of $\tau_2 <: \tau_3$:
  \begin{description}
    \item[Case S-Top:] $\tau_1 <: \top$ by S-Top.
    \item[Case S-Bot:] by bottom-inversion (Lemma \ref{lma:bottom_inversion}).
    \item[Case S-Refine:]
      have $\tau_2 = \beta ~r_2$ and $\tau_3 = \beta r_3$ and $r_2 <: r_3$; \\
      case analysis on the derivation tree of $\tau_1 <: \tau_2$:
      \begin{description}
        \item[Subcase S-Refine:]
          have $\beta ~ r_1 <: \beta ~ r_2$ and $\beta ~ r_2 <: \beta ~ r_3$; \\
          obtain $r_1 <: r_3$ by mutual induction transitivity of refinement subtyping; \\
          finally $\beta ~ r_1 <: \beta ~ r_3$ by S-Refine.
        \item[Subcase S-NameUp:]
          have $\tau_1 = n_1~r$ and $\tau_2 = \beta ~ r_2 = n_2~r_2$ and $\tau_3 = \beta~r_3=n_2~r_3$ and $n_1 <: n_\Sigma \in \Sigma$ and $n_\Sigma ~ r_1 <: n_2 ~ r_2$; \\
          obtain $n_\Sigma~r_1 <: n_3~r_3$ by induction hypothesis with smaller left subtyping derivation ($n_\Sigma~r_1 <: n_2~r_2$) and original right subtyping derivation ($n_2~r_2 <: n_2~r_3$); \\
          finally $n_1 ~r_1 <: n_2~r_3$ by S-NameUp.
        \item[Subcase S-Lower:]
          have $\tau_1 = p.t~r_1$ and $\JGSUpcast{p.t~r_1}{\tau_1'}$ and $\tau_1' <: \beta~r_2$; \\
          obtain $\tau_1' <: \beta~r_3$ by induction hypothesis with smaller left subtyping derivation ($\tau_1' <: \beta~r_2$) and original right subtyping derivation ($\beta~r_2 <: \beta~r_3$); \\
          finally $p.t~r_1 <: \beta~r_3$ by S-Lower.
        \item[Subcase S-Upper:]
          have $\tau_2 = \beta~r_2 = p.t~r_2$ and $\JGSDowncast{p.t~r_2}{\tau_2'}$ and $\tau_1 <: \tau_2'$ and $r_2 <: r_3$; \\
          obtain $\tau_3'$ such that $\JGSDowncast{p.t~r_3}{\tau_3'}$ and $\tau_1 <: \tau_3'$ by downcast-refinement (Lemma \ref{lma:downcast_refinement}); \\
          finally $\tau_1 <: p.t~r_3$ by S-Upper.
        \item[Subcases S-Top, S-Bot:] impossible.
      \end{description}

    \item[Case S-NameUp:]
    have $\tau_2 = n_2~r_2$ and $\tau_3 = n_3~r_3$ and $n_2~r_\Sigma <: n_\Sigma \in \Sigma$ and $r_2 <: r_\Sigma$ and $n_\Sigma~r_2 <: n_3~r_3$; \\
    case analysis on the derivation tree of $\tau_1 <: \tau_2$:
    \begin{description}
      \item[Subcase S-Refine:]
        have $n_2 ~ r_1 <: n_2 ~ r_2$ and $r_1 <: r_2$; \\
        obtain $n_\Sigma ~ r_1 <: n_3 ~ r_3$ by induction hypothesis with smaller right subtyping derivation ($n_\Sigma~r_2 <: n_3~r_3$);
        finally $n_2~r_1 <: n_3~r_3$ by S-NameUp.
      \item[Subcase S-NameUp:]
        have $\tau_1 = n_1~r_1$ and $\tau_2 = \beta ~ r_2 = n_2~r_2$ and $\tau_3 = \beta~r_3=n_2~r_3$ and $n_1 <: n_\Sigma' \in \Sigma$ and $r_1 <: r_\Sigma'$ and $n_\Sigma' ~ r_1 <: n_2 ~ r_2$; \\
        obtain $n_\Sigma'~r_1 <: n_3~r_3$ by induction hypothesis with smaller left subtyping derivation ($n_\Sigma'~r_1 <: n_2~r_2$) and original right subtyping derivation ($n_2~r_2 <: n_2~r_3$); \\
        finally $n_1 ~r_1 <: n_2~r_3$ by S-NameUp.
      \item[Subcase S-Lower:]
        have $\tau_1 = p.t~r_1$ and $\JGSUpcast{p.t~r_1}{\tau_1'}$ and $\tau_1' <: n_2~r_2$; \\
        obtain $\tau_1' <: n_2~r_3$ by induction hypothesis with smaller left subtyping derivation ($\tau_1' <: n_2~r_2$) and original right subtyping derivation ($n)2~r_2 <: n_3~r_3$); \\
        finally $p.t~r_1 <: n_3~r_3$ by S-Lower.
      \item[Subcases S-Top, S-Bot, S-Upper:] impossible.
  \end{description}

  \item[Case S-Lower:]
    have $\tau_2 = p.t~r_2$ and $\JGSUpcast{p.t~r_2}{\tau_2'}$ and $\tau_2' <: \tau_3$; \\
    obtain $\tau_1 <: \tau_2'$ by upcast-is-super (Lemma \ref{lma:upcast_is_super}); \\
    finally $\tau_1 <: \tau_3$ by induction hypothesis with smaller right subtyping derivation ($\tau_2' <: \tau_3$).

  \item[Case S-Upper:]
    have $\tau_3 = p.t~r_3$ and $\JGSDowncast{p.t~r_3}{\tau_3'}$ and $\tau_2 <: \tau_3'$; \\
    obtain $\tau_1 <: \tau_3'$ by induction hypothesis with smaller right subtyping derivation ($\tau_2 <: \tau_3'$); \\
    finally $\tau_1 <: p.t~r_3$ by S-Upper.
  \end{description}

  For the transitivity of refinement subtyping, we show that for refinements $r_1$, $r_2$ and $r_3$, if $r_1 <: r_2$ and $r_2 <: r_3$, then $r_1 <: r_3$.
  We start by induction over the derivation $r_2 <: r_3$:
  \begin{description}
    \item[Case S-R-Nil:] $r_1 <: {}$ by S-R-Nil.
    \item[Case S-R-Cons:]
      have $r_3 = \{ \Type{t}{B_3}{\tau_3}, \overline{\delta} \}$ and $\Type{t}{B_2}{\tau_2} \in r_2$ and $\Type{t}{B_2}{\tau_2} <: \Type{t}{B_3}{\tau_3}$ and $r_2 <: \{\overline{\delta}\}$; \\
      obtain $B_1$, $\tau_1$ such that $\Type{t}{B_1}{\tau_1} <: \Type{t}{B_2}{\tau_2}$ by refinement-member-inversion (Lemma \ref{lma:refinement_member_inversion}); \\
      finally $r_1 <: r_3$ by induction on smaller right refinement derivation ($r_2 <: \{\overline{\delta}\}$).
  \end{description}
\end{proof}

\begin{lemma}[Reflexivity of subtyping]
  Subtyping is reflexive for types and refinements: $\tau <: \tau$ and $r <: r$.
\end{lemma}
\label{lma:subtyping_reflexive}

\begin{proof}
  By mutual induction on the subtyping and refinement derivations.
\end{proof}

\begin{lemma}[Bottom-inversion of subtyping]
  For a well-formed context $\Gamma$ and correspondingly well-formed types $\tau_1$ and $\tau_2$,
  if $\tau_1 <: \bot$ then $\tau_1 <: \tau_2$.
\end{lemma}
\label{lma:bottom_inversion}

\begin{proof}
  Straightforward induction on the subtyping derivation tree.
\end{proof}

\begin{lemma}[Refinement-member inversion]
  For a refinement subtyping $r_1 <: r_2$,
  and a refined type $\Type{t}{B_2}{\tau_2} \in r_2$,
  then there exists some $B_1$ and $\tau_1$ such that
  $\Type{t}{B_1}{\tau_1} <: \Type{t}{B_2}{\tau_2}$.
\end{lemma}
\label{lma:refinement_member_inversion}

\begin{proof}
  Straightforward induction on the refinement derivation tree.
\end{proof}

\begin{lemma}[Refinement addition]
  If $r_1 <: r_2$ and $r_1 <: r_L$ and $r_2 <: r_L$
  then $(r_L +_r r_1) <: (r_L +_r r_2)$.
\end{lemma}
\label{lma:refinement_addition}

\begin{lemma}[Downcast refinement]
  If $\JGSDowncast{p.t~r}{\tau'}$ and
     $\tau <: \tau'$ and
     $r <: r'$
  then there exists some $\tau''$ such that
     $\JGSDowncast{p.t~r'}{\tau''}$ and
     $\tau <: \tau''$.
\end{lemma}
\label{lma:downcast_refinement}
\begin{proof}
  By induction on $\tau <: \tau'$ using auxiliary lemmas refinement-addition (Lemma \ref{lma:refinement_addition}) and exposure-is-decidable (\ref{lma:exposure-is-decidable}).
\end{proof}

\begin{lemma}[Upcast is super (semitransitive)]
  If $\JGSUpcast{\tau_2}{\tau_2'}$ and
     $\tau_1 <: \tau_2$
  then
     $\tau_1 <: \tau_2'$.
\end{lemma}
\label{lma:upcast_is_super}
\begin{proof}
  By induction on $\tau_1 <: \tau_2$; requires auxiliary lemma refinement-addition (Lemma \ref{lma:refinement_addition}).
\end{proof}



\subsection{Type Safety: Progress}

To complete the proof of type safety, we prove progress on a fuel-annotated big-step semantics.

\begin{figure}[h]
  \input{A2-reduction-fuel}
  \caption{Fuel-annotated Reduction Rules}
  \label{fig:reduction-fuel}
\end{figure}

\autoref{fig:reduction-fuel} shows the reduction rules from \autoref{fig:reduction} with added annotations. The judgment form $\StepsToN{\mu}{e}{n}{\mu'}{l_\bot}$ denotes that the expression $e$ can evaluate to a depth of at most $n$, resulting in a location $l$ or $\bot$ if the operation fails to reduce to a value in the given depth limit.

The rule {\sc E-F-Stuck} states that evaluating $e$ when there is no remaining fuel results in a ``stuck'' expression ($\bot$).

The rule {\sc E-F-Loc} extends the original rule {\sc Ev-Loc} with a nonzero fuel annotation $n+1$.
Rules {\sc E-F-Field} and {\sc E-F-New} similarly require a nonzero fuel.

The rule {\sc E-F-Method} requires a nonzero fuel and decrements the fuel for the nested evaluation. The result of the nested evaluation may be a location or a bottom.

Rule {\sc E-F-Let} tries to evaluate a let-expression $\Let{x}{e_1}{e_2}$.
The rule applies when $e_1$ evaluates to a non-bottom location.
The evaluation of $e_2$ may result in a real location or a bottom.
Note that both nested evaluations run with the same decremented fuel $n$: thus, the fuel effectively limits the depth of evaluation, rather than limiting the number of evaluation steps.

Rule {\sc E-F-Let-Stuck} applies when the first expression in a let-expression is stuck.

We can now state progress theorem on this semantics.
Informally, we want to say that any well-typed closed expression will either produce a value, or explicitly run out of fuel (the expression will be \emph{stuck}).

In the following proof, we assume that a corresponding preservation theorem holds for the fuel-annotated semantics; translating between the two semantics is straightforward.

\begin{theorem}[Type safety: progress]
  Given a well-typed context $S$ and a correspondingly-well-typed expression $e$ ($\JType{\cdot}{S}{e}{\tau}$) and heap $\mu$ ($\JJudge{\cdot}{\mu}{S}$), then for a given fuel $n$, there exist some result heap $\mu'$ and optional result location $l_\bot$ such that $\StepsToN{\mu}{e}{n}{\mu'}{l_\bot}$.
\end{theorem}
\begin{proof}
  Induction on typing derivation $\JType{\cdot}{S}{e}{\tau}$:
  \begin{description}
    \item[When] $n = 0$:\\
      In all cases when fuel is zero, apply rule {\sc E-F-Stuck}.
      In other cases, assume $n > 0$.
    \item[Case] {\sc T-Var}:\\
      Impossible (variable context is empty).
    \item[Case] {\sc T-Loc}:\\
      Apply rule {\sc E-F-Loc}.
    \item[Case] {\sc T-Let}:\\
      Have $\Let{x}{e_1}{e_2}$ and $\JType{\cdot}{S}{e_1}{\tau_1}$ and $\JType{x: \tau_1}{S}{e_2}{\tau_2}$ and $\JAvoid{x: \tau_1}{\tau_2}{\le}{x}{\tau}$. \\
      By inductive hypothesis, have $\StepsToN{\mu}{e_1}{n-1}{\mu'}{l_{1\bot}}$. \\
      Nested case distinction on $l_{1\bot}$:
      \begin{description}
        \item[Subcase] $\bot$: \\
          Fuel exhausted; apply rule {\sc E-F-Let-Stuck}.
        \item[Subcase] $l_1$: \\
          By lemma preservation of reduction, have extended context $S'$ such that $l_1: \tau_1$. \\
          By substitution and weakening of $e_2$'s typing judgment, have $\JType{}{S'}{e_2[x := l_1]}{\tau_2[x := l_1]}$. \\
          By inductive hypothesis, have $\StepsToN{\mu'}{e_2}{n-1}{\mu''}{l_{2\bot}}$. \\
          Apply rule {\sc E-F-Let}.
    \end{description}
    \item[Case] {\sc T-Sel}:\\
      Have $e = p.v$. \\
      By assumption that $p.v$ well-typed in empty context and that heap $\mu$ well-typed, $p$ is a location and field $v$ exists in $\mu(p)$. \\
      Apply rule {\sc E-F-Field}.
    \item[Case] {\sc T-App}:\\
      Have $e = p_s.f(p_a)$. \\
      By assumption that $p_s.f(p_a)$ well-typed in empty context and that heap $\mu$ well-typed, $p_s$ is a location and method body $e_f$ exists in $\mu(p_s)$. \\
      By well-typed heap assumption, method body $e_f$ is well-typed under context $x_s: \tau_s, x_a: \tau_a$. \\
      By repeated substitution, $e_f[x_s := l_s, x_a := l_a]$ well-typed under empty context. \\
      By inductive hypothesis, $e_f[x_s := l_s, x_a := l_a]$ steps with fuel $n-1$ to some heap $\mu'$ and location-or-stuck $l'_\bot$. \\
      Apply rule {\sc E-F-App}.
    \item[Case] {\sc T-New}:
      Apply rule {\sc E-F-New} with fresh location $l$.
  \end{description}
\end{proof}

\end{document}